\documentclass[sigconf]{acmart}
\renewcommand\footnotetextcopyrightpermission[1]{} % removes footnote with conference information in first column
\setcopyright{none}

\usepackage{hyperref}

\usepackage{array}
\usepackage{url}

\usepackage[T1]{fontenc}   
\usepackage{graphicx}
\usepackage{subfig}
\usepackage{caption}

\usepackage{amsmath}
\usepackage{courier}
\usepackage{listings}
\usepackage{xcolor}
\usepackage{epstopdf}

\usepackage[linesnumbered,ruled,vlined]{algorithm2e}
\SetAlFnt{\scriptsize}
\SetAlCapFnt{\scriptsize}
\SetAlCapNameFnt{\scriptsize}

\usepackage[]{algorithm2e}
\usepackage{mathtools}
\usepackage{pgfplots}

\usepackage{enumitem}
\usepackage{booktabs}
\usepackage{multirow}
\usepackage{rotating}
\usepackage{makecell}

\usepackage{MnSymbol}
\newcommand{\htor}[1]{}
\newcommand{\maciej}[1]{}

\newcommand{\goal}[1]{}

\newcommand{\flib}{\textsc{ftRMA}}
\newcommand{\fompi}{\textsc{foMPI}}

\newcommand{\ptoq}{$(p \rightarrow q)$}

\newcommand{\ptoqn}{$p \rightarrow q$}

\newcommand{\Dptoq}{$(p \mathrel{\substack{\textstyle\Rightarrow\\[-0.3ex]
                      \textstyle\rightarrow}} q)$}
\newcommand{\Dpfromq}{$(p \mathrel{\substack{\textstyle\Leftarrow\\[-0.3ex]
                      \textstyle\rightarrow}} q)$}
\newcommand{\Dqtop}{$(q \mathrel{\substack{\textstyle\Rightarrow\\[-0.3ex]
                      \textstyle\rightarrow}} p)$}

\newcommand{\Dqtor}{$(q \mathrel{\substack{\textstyle\Rightarrow\\[-0.3ex]
                      \textstyle\rightarrow}} r)$}

\newcommand{\Drtoq}{$(r \mathrel{\substack{\textstyle\Rightarrow\\[-0.3ex]
                      \textstyle\rightarrow}} q)$}

\newlength{\verticalcompensationlength}
\newcounter{verticalcompensationrows}
\newcommand{\verticalcompensation}[1]{%
\setcounter{verticalcompensationrows}{#1}%
\addtocounter{verticalcompensationrows}{-1}%
\vspace*{-\value{verticalcompensationrows}\verticalcompensationlength}%
}

\newcommand{\multirowbt}[3]{%
\multirow{#1}{#2}{\verticalcompensation{#1}#3}%
}

\newcommand{\smalltt}[1]{{\small\texttt{#1}}}

\usepackage{mathrsfs}

\usepackage[utf8]{inputenc}
\usepackage{cleveref}
\crefname{section}{§}{§§}
\Crefname{section}{§}{§§}

\newtheorem{defi}{Definition}
\newtheorem{theorem}{Theorem}[section]

\begin{document}

\pgfplotscreateplotcyclelist{listForPlots}{%
%{},
{mark=star},
{blue,mark=square*},
{black,mark=otimes*},
{brown,mark=diamond*},
{red,mark=triangle*},
{black,mark=pentagon},
{orange,mark=oplus}}

\pgfplotscreateplotcyclelist{listForPlotsNoMarkers}{%
%{},
{red,mark=none,style=solid, thick},
{blue,mark=none,style=dashed, thick},
{black,mark=none,style=densely dotted, thick},
{brown,mark=none,style=dashdotted, thick},
{black,mark=none,style=loosely dotted, thick}}

\title[Fault Tolerance for Remote Memory Access Programming Models]{Fault Tolerance for Remote Memory Access\\ Programming Models}

\author{Maciej Besta}
       \affiliation{Department of Computer Science\\
       ETH Zurich}
       \email{maciej.besta@inf.ethz.ch}
\author{Torsten Hoefler}
       \affiliation{Department of Computer Science\\
       ETH Zurich}
       \email{htor@inf.ethz.ch}

\begin{abstract}
%%%%% Motivation/problem statement: Why do we care about the problem? What
%%%%%practical, scientific, theoretical or artistic gap is your research
%%%%%filling?
%
%Recent developments in the architecture of large datacenters and HPC centers
%gave rise to new hardware features and software concepts. Examples are 
Remote Memory Access (RMA) is an emerging mechanism for programming
high-performance computers and datacenters. However, little work exists on
resilience schemes for RMA-based applications and systems.
%, which directly utilize the underlying communication RDMA-capable hardware,
%and growing hardware heterogeneity and complexity.  which requires new fault
%tolerance mechanisms for efficient utilization.
%
%%%%% Methods/procedure/approach: What did you actually do to get your
%%%%%results? (e.g. analyzed 3 novels, completed a series of 5 oil paintings,
%%%%%interviewed 17 students)
In this paper we analyze fault tolerance for RMA and show that it is
fundamentally different from resilience mechanisms targeting the message passing (MP)
model. We design a model for reasoning about fault tolerance for RMA,
addressing both flat and hierarchical hardware.  We use this model to construct
several highly-scalable mechanisms that provide efficient low-overhead
in-memory checkpointing, transparent logging of remote memory accesses, and a
scheme for transparent recovery of failed processes.
%fault tolerance mechanisms for RMA models in large hierarchical
%high-performance computers and datacenters. We design We also construct a
%reliability model of single and simultaneous failures for hierarchical
%hardware, which we use to further increase the resilience of our techniques.
%
%%%%% Results/findings/product: As a result of completing the above
%%%%%procedure, what did you learn/invent/create?
Our protocols take into account diminishing amounts of memory per core, one of major
features of future exascale machines. The implementation of our fault-tolerance
scheme entails negligible additional overheads. Our reliability model shows
that in-memory checkpointing and logging provide high resilience.
%%%%% Conclusion/implications: What are the larger implications of your %%%%
%%%%%findings, especially for the problem/gap identified in step 1?
This study enables highly-scalable resilience mechanisms for RMA and fills a
research gap between fault tolerance and emerging RMA programming
models. 
%and emerging hierarchical hardware
\end{abstract}


\maketitle
\pagestyle{plain}

{\vspace{-0.5em}\noindent \textbf{This is a full version of a paper published at\\ ACM HPDC'14 under the same title}}

{\vspace{1em}\small\noindent\textbf{Project website:}\\\url{https://spcl.inf.ethz.ch/Research/Parallel\_Programming/ftRMA}}

\section{Introduction}

\goal{Introduce RMA \& RDMA; convince readers they're getting common}

Partitioned Global Address Space (PGAS), and the wider class of Remote
Memory Access (RMA) programming models enable high-performance
communications that often outperform Message
Passing~\cite{fompi-paper,Petrovic:2012:HRB:2312005.2312029}. 
%
%Remote Memory Access (RMA) and similar Partitioned Global Address Space
%(PGAS) programming models enable high-performance one-sided communication in HPC
%centers and datacenters, often outperforming the traditional message passing (MP)
%model~\cite{fompi-paper,Petrovic:2012:HRB:2312005.2312029}.  
RMA utilizes remote direct memory access (RDMA) hardware features to
access memories at remote processes without involving the OS or the
remote CPU.

RDMA is offered by most modern HPC networks (InfiniBand, Myrinet, Cray's Gemini
and Aries, IBM's Blue Gene, and PERCS) and many Ethernet interconnects
that use the RoCE or iWARP protocols. RMA languages and
libraries include Unified Parallel C (UPC), Fortran 2008 (formerly
known as CAF), MPI-3 One Sided, Cray's SHMEM
interface, or Open Fabrics (OFED).  Thus,
we observe that RMA
is quickly emerging to be the programming model of choice for cluster systems, HPC computers, and
large datacenters.

\goal{Introduce basic concepts (CR, ML) in FT and motivate our work}

%%%%Fault tolerance of such systems is important because hardware
%%%%and software faults are ubiquitous~\cite{Sato:2012:DMN:2388996.2389022}. 
%%%%Two popular resilience schemes used in today's computing environments
%%%%are checkpoint/restart
%%%%(CR) and message logging (ML)~\cite{Elnozahy:2002:SRP:568522.568525}
%%%%\htor{sounds wrong to me, it's really uuncoordinated which requires ML
%%%%and coordinated CR which requires network silence. What you call ML also
%%%%needs checkpoints!!}. 

\sloppy
Fault tolerance of such systems is important because hardware
and software faults are ubiquitous~\cite{Sato:2012:DMN:2388996.2389022}. 
Two popular resilience schemes used in today's computing environments
are coordinated checkpointing (CC) and uncoordinated checkpointing augmented
with message logging (UC)~\cite{Elnozahy:2002:SRP:568522.568525}.
In CC applications regularly synchronize to save
their state to memory, local disks, or parallel file system (PFS)~\cite{Sato:2012:DMN:2388996.2389022};
this data is used to restart after a crash. In UC processes take checkpoints
independently and use message logging to avoid rollbacks caused by the \emph{domino effect}~\cite{Riesen:2012:ASI:2388996.2389021}. There has been considerable
research on CC and UC for the message passing (MP)
model~\cite{Elnozahy:2002:SRP:568522.568525,Alvisi:1998:MLP:630821.631222}.
Still, no work addresses the exact design of these schemes for
RMA-based systems.

\begin{table*}
\centering
\scriptsize \begin{tabular}{lllll}
\toprule

\parbox{0.1cm}{} & \textbf{MPI-3 one sided operation} & \textbf{UPC operation}
& \textbf{Fortran 2008 operation} & \textbf{Cat.} \\ \midrule
\multirowbt{2}{*}{\parbox{0.1cm}{\begin{turn}{90}comm.\end{turn}}} &
\parbox{6.2cm}{\textsf{MPI\_Put}, \textsf{MPI\_Accumulate},
\textsf{MPI\_Get\_accumulate},\\ \textsf{MPI\_Fetch\_and\_op},
\textsf{MPI\_Compare\_and\_swap}} & \parbox{4.9cm}{\textsf{upc\_memput},
\textsf{upc\_memcpy}, \textsf{upc\_memset},\\assignment (\textsf{=}), all UPC
collectives} & \parbox{3.8cm}{assignment (\textsf{=})} &
\parbox{0.7cm}{\textsc{put}} \\ \cmidrule{2-5}

\parbox{0.1cm}{} & \parbox{6.2cm}{\textsf{MPI\_Get}, \textsf{MPI\_Compare\_and\_swap},\\
\textsf{MPI\_Get\_accumulate}, \textsf{MPI\_Fetch\_and\_op}} & \parbox{4.9cm}{\textsf{upc\_memget},
\textsf{upc\_memcpy}, \textsf{upc\_memset},\\assignment (\textsf{=}), all UPC
collectives} & \parbox{3.8cm}{assignment (\textsf{=})} &
\parbox{0.7cm}{\textsc{get}} \\ \midrule

\multirowbt{4}{*}{\parbox{0.1cm}{\begin{turn}{90}sync.\end{turn}}} &
\parbox{6.2cm}{\textsf{MPI\_Win\_lock}, \textsf{MPI\_Win\_lock\_all}} & \parbox{4.9cm}{\textsf{upc\_lock}} & \parbox{3.8cm}{\textsf{lock}} & \parbox{0.7cm}{\textsc{lock}}\\
\cmidrule{2-5}

\parbox{0.1cm}{} & \parbox{6.2cm}{\textsf{MPI\_Win\_unlock},
\textsf{MPI\_Win\_unlock\_all}} & \parbox{4.9cm}{\textsf{upc\_unlock}} &
\parbox{3.8cm}{\textsf{unlock}} & \parbox{0.7cm}{\textsc{unlock}}\\
\cmidrule{2-5}

\parbox{0.1cm}{} & \parbox{6.2cm}{\textsf{MPI\_Win\_fence}} &
\parbox{4.9cm}{\textsf{upc\_barrier}} & \parbox{3.8cm}{\textsf{sync\_all},
\textsf{sync\_team}, \textsf{sync\_images}} &
\parbox{0.7cm}{\textsc{gsync}}\\ \cmidrule{2-5}

\parbox{0.1cm}{} & \parbox{6.2cm}{\textsf{MPI\_Win\_flush},
\textsf{MPI\_Win\_flush\_all}, \textsf{MPI\_Win\_sync}} &
\parbox{4.9cm}{\textsf{upc\_fence}} & \parbox{3.8cm}{\textsf{sync\_memory}} &
\parbox{0.7cm}{\textsc{flush}}\\ \bottomrule
\end{tabular}
%\vspace{-0.5em}
\caption{Categorization of MPI One Sided/UPC/Fortran 2008 operations in our
model. Some atomic functions are considered as both \textsc{put}s and
\textsc{get}s. 
%In MPI,
%atomics explicitly transfer data from a target to a source and
%in the opposite direction. -- this is wrong, accumulates don't
In UPC, the collectives, assignments and
\textsf{upc\_memset}/\textsf{upc\_memcpy} behave similarly depending on
the values of pointers to shared objects; the same applies to Fortran
2008. We omit MPI's post-start-complete-wait synchronization 
and request-based RMA operations for simplicity.}
%As these actions are conceptually similar to the modeled operations, this
%omission does not affect the conclusions drawn.
%\vspace{-1.0em}
\label{tab:mpiCategories}
\end{table*}

In this work we develop a generic model for reasoning
about resilience in RMA. Then, using this model, we show that CC and UC for RMA fundamentally differ
from analogous schemes for MP. We also construct protocols that
enable simple checkpointing and logging of remote memory accesses. We \emph{only} use
\emph{in-memory} mechanisms to avoid costly I/O flushes and frequent
disk and PFS failures~\cite{Sato:2012:DMN:2388996.2389022, disk_fails}.
We then
extend our model to cover two features of today's petascale and future
exascale machines: (1) the growing complexity of hardware components and (2)
decreasing amounts of memory per core.  \emph{With this, our study fills an
important knowledge gap between fault-tolerance and emerging
RMA programming in large-scale computing systems.}

%\emph{With this, our study fills an
%important research gap between fault-tolerance mechanisms and emerging hardware
%\& software concepts to enable resilient RMA programming in large-scale
%computing systems.}

\goal{itemize and describe our concrete constributions}

In detail, we provide the following major contributions:
%
%\vspace{-0.5em}
\begin{itemize}[leftmargin=1em] 
\item We design a model for reasoning about the reliability of RMA systems
running on flat and hierarchical hardware with limited memory per
core. To our knowledge, this is the first work that addresses these
issues.
%\vspace{-0.5em}
\item We construct schemes for in-memory checkpointing, logging, and recovering RMA-based
applications.
%\vspace{-0.5em}
\item We unify these concepts in a topology-aware diskless protocol and
we use real data and an analytic model to show that the protocol
can endure concurrent hardware failures.
%\vspace{-1.6em}
\item We present the implementation of our protocol, analyze its performance,
show it entails negligible overheads, and compare it to other schemes.
\end{itemize}

%%%%%table with operations

%\vspace{-1.7em}
\section{RMA Programming} \label{sec:formalModel}

\goal{Introduce our model and describe memory sharing}

We now discuss concepts of RMA programming and present a formalization
that covers existing RMA/PGAS models with strict or relaxed
memory consistency (e.g., UPC or MPI-3 One Sided). 
In RMA, each process explicitly
exposes an area of its local memory as shared. Memory can be shared in different
ways (e.g., MPI windows, UPC shared arrays, or Co-Arrays in Fortran 2008); details are
outside the scope of this work. Once shared,
memory can be accessed with various language-specific operations.
%%A base principle in RMA is \emph{structured sharing} where each process explicitly
%%exposes an area of its local memory as shared. Memory can be shared in different
%%ways (e.g., MPI windows, UPC shared arrays or Co-Arrays in Fortran 2008); details are
%%outside the scope of this work. Once shared,
%%memory can be accessed with various language-specific actions.

%We model a distributed system as a
%tuple $(\mathcal{D},\mathcal{S})$, where $\mathcal{D}$ and $\mathcal{S}$ are
%sets containing, respectively, $P$ processes that communicate using only RMA,
%and all the data structures in the system.

%\subsection{Exposing Memories}

%as opposed to \emph{shared
%everything} in threaded programming models

%We denote shared and private memory at process $p$ as $SM_p$ and $PM_p$, respectively.
%$SM_p$ and $PM_p$ are sets of appropriate words in $q$'s memory.

%\vspace{-1.0em}
\subsection{RMA Operations} 
\label{sec:rma_ops}

\goal{+ Distinguish between communication/synchronization calls and active/passive procs}

\sloppy
We identify two fundamental types of RMA operations:
\emph{communication} actions (often called \emph{accesses}; they transfer data between processes), and
\emph{synchronization} actions (synchronize processes and guarantee memory
consistency). A process $p$ that issues an RMA action targeted at
$q$ is called the \emph{active} \emph{source}, and
$q$ is called the \emph{passive} \emph{target}. 
%With regards to a set of
%operations
%, a process can be: active, passive, or both (RMA enables a
%process to access remote memories and be accessed by other processes at
%the same time). 
We assume $p$ is active and
$q$ is passive (unless stated otherwise).

% active and passive 

%\vspace{-1.2em}
\subsubsection{Communication Actions}

\goal{++ Introduce puts/gets and a formalism for communication functions}

%Assume that process $p$ issues a communication action targeted at process $q$.

We denote an action that transfers data from $p$ to $q$ and from $q$ to $p$ as \textsc{put}\Dptoq\
and \textsc{get}\Dpfromq\@, respectively.
We use double-arrows to emphasize the asymmetry of the two
operations: the upper arrow indicates the direction of data flow and the lower arrow
indicates the direction of control flow. 
%If $q$ is active then we swap $p$ and $q$ in the notation
%(\textsc{put}\Dqtop\@, \textsc{get}\Dqfromp\@). htor - no content
%
%%%To exactly specify which regions of memory
%%%are modified, we use a notation \textsc{put}(\Dptoqn$, var = x$) (assign $x$ to $var$ in $q$'s memory)
%%%and \textsc{get}(\Dpfromqn$, var = x$) (read $x$ from $q$'s and assign to $var$ in $p$'s memory).
%We also use term \emph{access} to refer
%
The upper part of Table~\ref{tab:mpiCategories}
categorizes communication operations in various RMA languages. Some
actions (e.g., atomic compare and swap) transfer data in \emph{both} directions and thus 
fall into the family of \textsc{put}s \emph{and} \textsc{get}s.

\goal{++ Model atomics that combine remote and local data}

We also distinguish between \textsc{put}s that ``blindly'' replace a targeted memory region at $q$ with a new value
(e.g., UPC assignment), and \textsc{put}s that combine the data moved to
$q$ with the data that already resides at $q$ (e.g.,
\textsf{MPI\_Accumulate}). When necessary, we refer to the former type as the \emph{replacing} \textsc{put}, and to the latter as the \emph{combining} \textsc{put}.

%We use symbols $.args$ and
%$.data$ to refer to, respectively, arguments and data of any action (e.g.,
%\textsc{put}\Dptoq$.args$). 

%\vspace{-1.2em}
\subsubsection{Memory Synchronization Actions}

\goal{++ Describe and formalize synchronization actions}

We identify four major categories of memory synchronization actions:
\textsc{lock}(\ptoqn$,str)$ (locks a structure $str$ in $q$'s memory to
provide exclusive access),
\textsc{unlock}(\ptoqn$,str)$ (unlocks $str$ in $q$'s memory and
enforces consistency of $str$),
\textsc{flush}(\ptoqn$,str)$ (enforces consistency of $str$ in $p$'s and
$q$'s memories)\@, and \textsc{gsync}$(p \to \diamond, str)$ (enforces
consistency of $str$); $\diamond$ indicates that a call targets all processes.
Arrows indicate the flow of control (synchronization). When we refer to
the whole process memory (and not a single structure), we omit $str$
(e.g., \textsc{lock}(\ptoqn)). The lower part of
Table~\ref{tab:mpiCategories} categorizes synchronization calls in
various RMA languages.

\subsection{Epochs and Consistency Order}
\label{sec:epochs_consistency}

\goal{+ Explain and formalize epochs}

%A major feature of RMA programming, provided by nearly all RMA languages and
%libraries, is loose memory consistency. It enables issuing non-blocking
%\textsc{put}s and \textsc{get}s. Issued operations are completed by calling
%memory consistency functions (\textsc{flush}, \textsc{unlock},
%\textsc{gsync}). The period between any two such operations is called the
%\emph{epoch}. Every memory consistency call \emph{closes} the current epoch and
%\emph{opens} a new one. 

RMA's relaxed memory consistency enables non-blocking
\textsc{put}s and \textsc{get}s. Issued operations are completed by 
memory consistency actions (\textsc{flush}, \textsc{unlock},
\textsc{gsync}). The period between any two such actions issued by $p$ and targeting the same
process $q$ is called an
\emph{epoch}. Every \textsc{unlock}\ptoq\ or \textsc{flush}\ptoq\ 
\emph{closes} $p$'s current epoch and
\emph{opens} a new one (i.e., increments $p$'s epoch number denoted as
$E$\ptoq\@).
%\htor{if it's p's number, what does q do in the equation?}.
%
$p$ can be in several
independent epochs related to each process that it communicates with.
%Note that $q$ can also be in an
%independent epoch $E$\qtop\ incremented with an \textsc{unlock}\qtop\ or
%a \textsc{flush}\qtop\@. \htor{what does this add?}
%
As \textsc{gsync} is a collective call, it increases
epochs at every process.

An important concept related to epochs is the \emph{consistency order} (denoted as $\xrightarrow{co}$). $\xrightarrow{co}$ orders
the visibility of actions: $x \xrightarrow{co} y$ means that memory effects of action $x$ are globally visible before action $y$.
Actions issued in different epochs by process $p$ targeting the same process $q$ are always ordered with $\xrightarrow{co}$. Epochs and $\xrightarrow{co}$ are illustrated in Figure~\ref{fig:epochs}. $x\ ||_{co}\ y$ means that actions $x$ and $y$ are \emph{not} ordered with $\xrightarrow{co}$.

% Finally, we use symbol
%$.EA$ to refer to the active epoch in which a communication operation was
%issued (e.g., \textsc{put}\Dptoq$.EA$).

%\vspace{-0.5em}
\begin{figure}[h!] \centering
\includegraphics[width=0.44\textwidth]{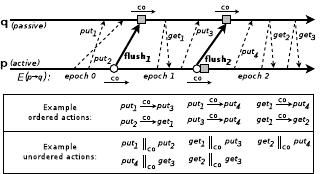}
%\vspace{-0.5em}
\caption{Epochs and the consistency order
  $\xrightarrow{co}$ (\cref{sec:epochs_consistency}). White circles
  symbolize synchronization calls (in this case \textsc{flush}).
Grey squares show when calls' results become
globally visible in $q$'s or $p$'s memory.
%Dashed arrows are any
%communication actions.
} \label{fig:epochs}
%\vspace{-1.0em}
\end{figure}

%
%Every time any process issues an $unlock$, $flush$, or $gsync$, both active
%and passive epoch numbers are incremented. 
%
%For a failure-free scenario, we
%have the following \emph{epoch equations}:
%
%%\vspace{-2.0em} 
%
%%\scriptsize 
%\begin{alignat}{2} EA(p \rightarrow q) &= EP(p\rightarrow q)\\
% EA(q \rightarrow p) &= EP(q \rightarrow p) 
%\end{alignat}
%%\normalsize
%
%\noindent where $p$ and $q$ are any two processes in the system.

%%\vspace{-0.5em}
%\begin{figure*}
%\centering
%\includegraphics[width=1.0\textwidth]{log_s-eps-converted-to.pdf}
%%\vspace{-1.5em}
%\caption{.}
%\label{fig:log_s}
%%\vspace{-0.5em}
%\end{figure*}

%\vspace{-1.1em}
\subsection{Program, Synchronization, and Happened Before Orders}
\label{sec:orders}

\goal{+ Introduce and explain PO, SO, HB orders}

In addition to $\xrightarrow{co}$ we require three more orders to
specify an RMA execution~\cite{hoefler2013remote}:
The \emph{program order} ($\xrightarrow{po}$) specifies the order of
actions of a single thread, similarly to the program order in
Java~\cite{Manson:2005:JMM:1040305.1040336} ($x \xrightarrow{po} y$
means that $x$ is called before $y$ by some thread).
%
%%%%%%%The \emph{synchronization order} ($\xrightarrow{so}$) provides the total order of \textsc{lock}s, \textsc{unlock}s, and \textsc{gsync}s {targeted} at {the same} process. %$x \xrightarrow{so} y$ means that $x$ and $y$ 
%
The \emph{synchronization order} ($\xrightarrow{so}$) orders
\textsc{lock} and \textsc{unlock} and other synchronizing operations.
\emph{Happened-before} (HB, $\xrightarrow{hb}$), a relation well-known in message passing~\cite{Lamport:1978:TCO:359545.359563}, is the transitive closure of the union of $\xrightarrow{po}$ and $\xrightarrow{so}$.
We abbreviate a \emph{consistent happen-before} as $\xrightarrow{cohb}$: $a \xrightarrow{cohb} b \equiv a \xrightarrow{co} b \wedge a \xrightarrow{hb} b$.
To state that actions are \emph{parallel} in an order, we use the symbols $||_{po},\ ||_{so},\ ||_{hb}$. 
We show the orders in
Figure~\ref{fig:orders}; more details can be found
in~\cite{hoefler2013remote}.

\begin{figure}[h!] \centering
\includegraphics[width=0.44\textwidth]{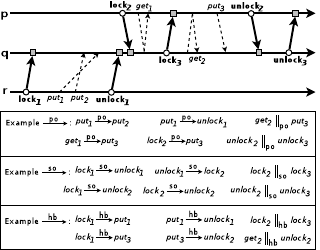}
%\vspace{-0.8em}
\caption{Example RMA orderings $\xrightarrow{po}, \xrightarrow{so}$, $\xrightarrow{hb}$ (\cref{sec:orders}).} \label{fig:orders}
%\vspace{-1.0em}
\end{figure}

%%\vspace{-0.5em}
%\begin{figure}[h!] \centering
%\includegraphics[width=0.48\textwidth]{orders-eps-converted-to.pdf}
%%\vspace{-0.5em}
%\caption{Illustration of example RMA orderings ($\xrightarrow{po}, \xrightarrow{so}$, $\xrightarrow{co}$)} \label{fig:orders}
%%\vspace{-1.0em}
%\end{figure}

%\vspace{1.2em}
\subsection{Formal Model}
\label{sec:formal_model}

\goal{+ Describe the model assumptions}

%We now formalize the above concepts.
We now combine the various RMA concepts and fault tolerance into a
single formal model.
We assume fail-stop faults (processes can disappear
nondeterministically but behave correctly while being a part of the
program). The data
communication may happen out of order as specified for most RMA models.
Communication channels between
non-failed processes are asynchronous, reliable, and error-free. 
The user code can only communicate and synchronize using RMA functions specified in Section~\ref{sec:rma_ops}.
Finally, checkpoints and logs are stored
in \emph{volatile} memories.

%Formally, we define a communication action as a tuple $\langle type, src, trg, dtrg, EA \rangle$. 
%$Type$ is the type of an action (a put or a get), $src$ and $trg$ determine the source and the target process, $dtrg$ is the
%process targeted by the data flow, and $EA$ is the active epoch number in which the operation was issued.
%We have \textsc{put}\Dptoq\ $\equiv \langle put, p, q, q, EA \rangle$, and \textsc{get}\Dpfromq\ $\equiv \langle get, p, q, p, EA \rangle$ (for any $EA$). Note that we introduce $trg$ and $dtrg$ to cover the differences in the data flow of puts and gets.
%
%Similarly, a synchronization action is defined as a tuple $\langle type, src, trg, EA, str \rangle$. We skip $dtrg$ (there are no data flows in this class of functions) and we add $str$ to refer to the structure targeted. We thus have: \textsc{lock}(\ptoqs$,str$\normalsize$) \equiv \langle lock, p, q, EA, str \rangle$, \textsc{unlock}(\ptoqs$,str$\normalsize$) \equiv \langle unlock, p, q, EA,\\str \rangle$, \textsc{flush}(\ptoqs$,str$\normalsize$) \equiv \langle flush, p, q, EA, str \rangle$, and \textsc{gsync}($str$\normalsize$) \equiv \langle gsync, \diamond, \diamond, EA, str \rangle$. We use $\diamond$ to denote all
%the processes in the system.
%
%We model a distributed system running an RMA program as a tuple $\mathcal{D} = \langle \mathcal{P}, \mathcal{A}, \mathcal{S}, \mathscr{S}, \xrightarrow{po}, \xrightarrow{so}, \xrightarrow{hb}, \xrightarrow{co} \rangle$, where: 

\goal{+ Formalize communication actions}

%\htor{we follow the notation of XXX}
We define a communication action $a$ as a tuple

%\vspace{-1.0em}
\small
\begin{alignat}{1}
a = \langle type, src, trg, combine, EC, GC, SC, GNC, data \rangle
\end{alignat}
\normalsize

\noindent
where $type$ is either a put or a get, $src$ and $trg$ specify the
source and the target, and $data$ is the data carried by $a$. $Combine$
determines if $a$ is a replacing \textsc{put} ($combine = false$) or a
combining \textsc{put} ($combine = true$). $EC$ (\emph{Epoch Counter})
is the epoch number in which $a$ was issued. $GC$, $SC$, and $GNC$ are
counters required for correct recovery; we discuss them in more detail
in Section~\ref{sec:complicated}. We combine the notation from
Section~\ref{sec:rma_ops} with this definition and write 
\textsc{put}\Dptoq$.EC$ to refer to the epoch in which the put happens.
We also define a {\emph{determinant}} of $a$ (denoted as $\#a$, cf.~\cite{Alvisi:1998:MLP:630821.631222}) to be tuple $a$ without $data$:

%\vspace{-1.0em}
\small
\begin{alignat}{1}
\#a = \langle type, src, trg, combine, EC, GC, SC, GNC \rangle.
\end{alignat}
\normalsize

%Note that $a = \#a || data$ where $||: ||(\langle a_1, ..., a_n \rangle, \langle b_1, ..., b_k \rangle = \langle a_1, ..., a_n, b_1, ..., b_k \rangle$ is the concatenation operator. 

\goal{+ Formalize synchronization actions}

\noindent
Similarly, a synchronization action $b$ is defined as

%\vspace{-1.0em}
\small
\begin{alignat}{1}
b = \langle type, src, trg, EC, GC, SC, GNC, str \rangle.
\end{alignat}
\normalsize

%(p \mathrel{\substack{\textstyle\Rightarrow\\[-0.3ex]\textstyle\rightarrow}} q)

\goal{+ Formalize a distributed system}

\noindent
Finally, a trace of an RMA program running on a distributed system can
be written as the
tuple

%(for example, $\textsc{put}$\Dptoq\ $ \equiv \langle put, p,\\q, true, EA, SE, PE, data \rangle$; the same refers to other actions).

%$\textsc{put}(p \mathrel{\substack{\textstyle\Rightarrow\\[-0.3ex]\textstyle\rightarrow}} q) \equiv \langle put, p, q, EA, data \rangle$,
%$\textsc{get}(p \mathrel{\substack{\textstyle\Leftarrow\\[-0.3ex]\textstyle\rightarrow}} q) \equiv \langle get, p, q, EA, data \rangle$,
%$\textsc{lock}(p \rightarrow q, str) \equiv \langle lock, p, q, EA, str \rangle$,
%$\textsc{unlock}(p \rightarrow q, str) \equiv \langle unlock, p, q, EA, str \rangle$,
%$\textsc{flush}(p \rightarrow q, str) \equiv \langle flush, p, q, EA, str \rangle$,
%$\textsc{gsync}(str) \equiv \langle gsync, \diamond, \diamond, EA, str \rangle$.

%\begin{alignat}{2}
%\textsc{put}(p \mathrel{\substack{\textstyle\Rightarrow\\[-0.3ex]\textstyle\rightarrow}} q) &\equiv \langle put, p, q, EA, data \rangle\\
%\textsc{get}(p \mathrel{\substack{\textstyle\Leftarrow\\[-0.3ex]\textstyle\rightarrow}} q) &\equiv \langle get, p, q, EA, data \rangle\\
%\textsc{lock}(p \rightarrow q, str) &\equiv \langle lock, p, q, EA, str \rangle\\
%\textsc{unlock}(p \rightarrow q, str) &\equiv \langle unlock, p, q, EA, str \rangle\\
%\textsc{flush}(p \rightarrow q, str) &\equiv \langle flush, p, q, EA, str \rangle\\
%\textsc{gsync}(str) &\equiv \langle gsync, \diamond, \diamond, EA, str \rangle
%\end{alignat}
%
%for any $EA$, $data$, and $str$. We use $\diamond$ to denote all the processes in the system.

%\mathcal{F},

%\vspace{-1.0em}
\small
\begin{alignat}{2}
\mathcal{D} = \langle \mathcal{P}, \mathcal{E},
\mathcal{S}, \xrightarrow{po}, \xrightarrow{so}, \xrightarrow{hb},
\xrightarrow{co} \rangle,
\end{alignat}
\normalsize
%\vspace{-1.0em}

\noindent
%where: 
where 
\begin{description}[leftmargin=1.4em]
  \itemsep-1pt
\item $\mathcal{P}$ is the set of all $\mathcal{P}$rocesses in
  $\mathcal{D}$ ($|\mathcal{P}| = N$),
\item $\mathcal{E} = \mathcal{A} \cup \mathcal{I}$ is
  the set of all $\mathcal{E}$vents: 
\item $\mathcal{A} $ is the set of RMA $\mathcal{A}$ctions,
\item $\mathcal{I}$ is the set of $\mathcal{I}$nternal
  actions (reads, writes, checkpoint actions). $\textsc{read}(x,p)$ loads local variable $x$ and
$\textsc{write}(x := val,p)$ assigns $val$ to $x$ (in $p$'s memory).
$C_{p}^{i}$ is the $i$th checkpoint action taken by $p$. Internal
  events are partially ordered with actions using $\xrightarrow{po}$,
  $\xrightarrow{co}$, and $\xrightarrow{hb}$.
\item $\mathcal{S}$ is the set of all data $\mathcal{S}$tructures used by
  the program
%\vspace{-0.7em}
%\item $\mathcal{F}$ is a set of several helper $\mathcal{F}$unctions (described later).
%\vspace{-0.7em}
%\item $\xrightarrow{po}$, $\xrightarrow{so}$, $\xrightarrow{co}$, and $\xrightarrow{hb}$ are orders described above. 

%\item $\xrightarrow{po}$ is a relation that specifies the program order of actions of a single thread, similarly to the ``program order'' in Java~\cite{Manson:2005:JMM:1040305.1040336}.
%\vspace{-0.7em}
%\item $\xrightarrow{so}$ is the total order of the synchronization relations.
%\vspace{-0.7em}
%\item $\xrightarrow{hb}$ a well-known happen-before relation~\cite{Lamport:1978:TCO:359545.359563}, in our model it is a transitive closure of $\xrightarrow{po}$ and $\xrightarrow{so}$.
%\vspace{-0.7em}
%\item $\xrightarrow{co}$ \maciej{as in htors}
%\item $\mathcal{L}$ is a set of all $\mathcal{L}$ogging actions: $log(\textsc{op})$, where $\textsc{op}$ is any issued \textsc{put} or \textsc{get}.
%\vspace{-0.7em}
%\item $\mathcal{E} = \mathcal{A} \cup \mathcal{I} \cup \mathcal{C} $. $\mathcal{A} $ is a set of all RMA $\mathcal{A}$ctions.
%\vspace{-0.7em}
%\item $\mathcal{I}$ is a set of all $\mathcal{I}$nternal events (reads and writes). $\textsc{read}(x)$ loads local variable $x$. $\textsc{write}(x = val)$ writes $val$ to $x$.
%\vspace{-0.7em}
\end{description}

\section{Fault-tolerance for RMA}
\label{sec:basicFTforRMA}

\goal{Introduce and summarize the section}

%We now present efficient mechanisms that utilize and integrate with RMA
%concepts to make emerging RMA programming models fault tolerant. Our mechanisms
%take advantage of the relaxed memory consistency and build on the epoch
%abstraction and model introduced before. In addition, we analyze
%the differences between RMA and MP.

We now present schemes that make RMA codes fault tolerant. We start with
the simpler CC and then present RMA protocols for UC.

%\maciej{any idea how to make it one-line?}
%we identify and describe differences between fault tolerance for RMA and MP and we

%%%%%%In this section we first identify and describe
%%%%%%differences between resilience for RMA and MP.
%%%%%%We then present mechanisms that make RMA applications fault tolerant.
%
%We assume that $p \in \mathcal{P}$ is the
%active and $q \in \mathcal{P}$ is the passive communication side.

%%\vspace{-1.1em}
%\subsection{MP vs. RMA: Fault Tolerance} 
%\label{sec:mp_vs_rma_ft}

%\vspace{-2.2em}
\subsection{Coordinated Checkpointing (CC)}
\label{sec:taking_coordinated_ckp}

%\goal{+ Say why we do coordinated and how we model it}

In many CC schemes, the user explicitly calls a function to take a
checkpoint. Such protocols may leverage RMA's features (e.g., direct
memory access) to improve the performance. However, these schemes have
several drawbacks: they complicate the code because they can only be
called when the network is quiet~\cite{632814} and they do not
always fit the optimality criteria such as Daly's checkpointing
interval~\cite{Daly:2006:HOE:1134241.1134248}.  In this section, we
first identify how CC in RMA differs from CC in MP and then describe a
scheme for RMA codes that performs CC \emph{transparently} to the
application.
We model a coordinated checkpoint as a set $C = \{C_{p_1}^{i_1},
C_{p_2}^{i_2}, ..., C_{p_N}^{i_N}\} \subseteq \mathcal{I}, p_m \neq p_n$
for any $m,n$.

%\vspace{-2.2em}
\subsubsection{RMA vs. MP: Coordinated Checkpointing}
\label{sec:rma_vs_mp_cc}

\goal{+ Explain why CC differ in MP and RMA}

In MP, every $C$ has to satisfy a \emph{consistency
condition}~\cite{632814}: $\forall C_{p}^{i}, C_{q}^{j} \in C:\
C_{p}^{i}\ ||_{hb}\ C_{q}^{j}$. This condition ensures that $C$ does not
reflect a system state in which one process received a message that was
\emph{not} sent by any other process. We adopt this condition and
extend it to cover all RMA semantics:

\begin{defi}
$C$ is RMA-consistent iff $\forall C_{p}^{i}, C_{q}^{j} \in C:\
C_{p}^{i}\ ||_{cohb}\ C_{q}^{j}$.
\end{defi}

We extend $||_{hb}$ to $||_{cohb}$ to
guarantee that the system state saved in $C$ does not contain a process
affected by a memory access that was \emph{not} issued by any other
process. In RMA, unlike in MP, this condition can be easily satisfied
because each process can drain the network with a local \textsc{flush}
(enforcing consistency at any point
is legal~\cite{hoefler2013remote})

\subsubsection{Taking a Coordinated Checkpoint}
\label{sec:cc_for_rma}

%We use diskless coordinated checkpoints to regularly clear the logs and provide a ``last resort'' in case the uncoordinated checkpointing and logging fail. 

\goal{+ Describe our CC schemes}

We now propose two diskless schemes that obey the
RMA-consistency condition and target MPI-3 RMA codes.  The first
(``Gsync'') scheme can be used in programs that \emph{only} synchronize with
\textsc{gsync}s. The other (``Locks'') scheme targets codes that
\emph{only} synchronize with \textsc{lock}s and \textsc{unlock}s. Note that in
correct MPI-3 RMA programs \textsc{gsync}s and
\textsc{lock}s/\textsc{unlock}s cannot be mixed~\cite{mpi3}. All our schemes
assume that a \textsc{gsync} may also introduce an additional
$\xrightarrow{hb}$ order, which is true in some
implementations~\cite{mpi3}.

\goal{+ Describe the ``gsyncs'' coordinated scheme}

\textbf{The ``Gsync'' Scheme }
Every process may take a coordinated checkpoint right after the user
calls a \textsc{gsync} and before any further RMA calls by: (1)
optionally enforcing the global $\xrightarrow{hb}$ order with an
operation such as \textsf{MPI\_Barrier} (denoted as \textsc{bar}), and
taking the checkpoint. 
Depending on the application needs, not every \textsc{gsync} has to be
followed by a checkpoint. We use Daly's
formula~\cite{Daly:2006:HOE:1134241.1134248} to compute the best
interval between such checkpoints and we take checkpoints after the right
\textsc{gsync} calls.

\goal{+ Prove that ``gsyncs'' satisfies the consistency condition}

%%%\newtheorem{gsync}{The Gsync scheme satisfies the RMA-consistency condition and does not deadlock.}
%%%\gsync

\begin{theorem}
The Gsync scheme satisfies the RMA-consistency condition and does not deadlock.
\end{theorem}

\htor{you need to define what the RMA consistency condition is,
formally! ``Definition 1: RMA Consistency'' ...}

\begin{proof}
We assume correct MPI-3 RMA programs represented by their trace $\mathcal{D}$~\cite{hoefler2013remote,mpi3}. For
all $p,q \in \mathcal{P}$, each $\textsc{gsync}(p \to \diamond)$ has a
matching $\textsc{gsync}(q \to \diamond)$ such that $[\textsc{gsync}(p
\to \diamond)\ ||_{hb}\ \textsc{gsync}(q \to \diamond)]$.  Thus, if
every process calls \textsc{bar} right after \textsc{gsync} 
then \textsc{bar} matching is guaranteed and the program cannot deadlock. In addition, the \textsc{gsync} calls
introduce a global consistency order $\xrightarrow{co}$ such that the
checkpoint is coordinated and consistent. 
\end{proof}

\goal{+ Describe the ``Locks'' coordinated scheme}

\textbf{The ``Locks'' Scheme }
Every process $p$ maintains a local \emph{Lock Counter} $LC_p$
that starts with zero and is incremented after each \textsc{lock} and
decremented after each \textsc{unlock}. When $LC_p = 0$, process $p$
can perform a checkpoint in three phases: (1) enforce consistency with a
\textsc{flush}$(p \to \diamond)$, (2) call a \textsc{bar} to provides
the global $\xrightarrow{hb}$ order, and (3) take a 
checkpoint $C_{p}^{i}$.  
%The decision whether to checkpoint can be based on
%some application-specific conditions or, e.g., loosely synchronized
%clocks~\cite{Tong:1992:RRD:628900.629082}. 
The last phase, the actual checkpoint stage, is performed collectively
thus all processes can take the checkpoint $C$ in coordination.

%If synchronization is limited to locks, every process $p$ can freely issue a \textsc{gsync}$(p \to \diamond)$ followed by $C_{p}^{i}$. However, it has to first leave all the local critical sections (i.e., call $\textsc{unlock}(p \rightarrow q)$ for all $q$). Eventually, every process $p \in \mathcal{P}$ will take a checkpoint.

\goal{+ Prove that in ``Locks'' a checkpoint is always eventually taken}

\begin{theorem}
The Locks scheme satisfies the RMA-consistency condition and does not deadlock.
\end{theorem}

\begin{proof}
The call to \textsc{flush}$(p \to \diamond)$ in phase 1 guarantees global
consistency at each process. The \textsc{bar} in phase 2 guarantees that
all processes are globally consistent before the checkpoint taken in phase~3.

It remains to proof deadlock-freedom.
We assume correct MPI-3 RMA programs~\cite{hoefler2013remote,mpi3}. 
A $\textsc{lock}(p \to q)$ can only block waiting for an active lock
$\textsc{lock}(z \to q)$ and no \textsc{bar} can be started at $z$ while
the lock is held.
In addition, for
every $\textsc{lock}(z \to q)$, there is a matching $\textsc{unlock}(z
\to q)$ in the execution such that $\textsc{lock}(z \to q)
\xrightarrow{po} \textsc{unlock}(z \to q)$ (for any $z,p,q \in
\mathcal{P}$).  
Thus, all locks must be released eventually, i.e., $\exists a \in
\mathcal{E}:\ a \xrightarrow{po} \textsc{write}(LC_p := 0,p)$ for any $p
\in \mathcal{P}$.
\end{proof}

The above schemes show that the transparent CC can be achieved much
simpler in RMA than in MP.  In MP, such protocols usually have to
analyze inter-process dependencies due to sent/received messages, and
add protocol-specific data to
messages~\cite{Elnozahy:2002:SRP:568522.568525,Chandy:1985:DSD:214451.214456},
which reduces the bandwidth.  In RMA this is not necessary.

%%%%
%%%% A generic protocol
%%%%that spans the whole RMA semantics is more complicated\footnotemark[1],
%%%%however here we focus on practical applications that use only one synchronization mechanism at a time
%%%%as advised by, e.g., the MPI-3 RMA specification~\cite{mpi3}.

%\vspace{-0.5em}
\subsection{Uncoordinated Checkpointing (UC)}
\label{sec:uncoordinated_ckp}

Uncoordinated checkpointing augmented with message logging
reduces energy consumption and synchronization costs
because a single process crash does not force all other processes 
to revert to the previous checkpoint and recompute~\cite{Riesen:2012:ASI:2388996.2389021,Elnozahy:2002:SRP:568522.568525}.
Instead, a failed process fetches its last checkpoint and replays messages
logged beyond this checkpoint. However, UC schemes are usually more
complex than CC~\cite{Elnozahy:2002:SRP:568522.568525}.
We now analyze how UC in RMA differs from UC in MP, followed by a
discussion of our UC protocols.

%%%%Logging of accesses in RMA brings benefits similar to recording messages
%%%%in MP. For example, if no logging is enabled, a single process crash forces all other processes 
%%%%to rollback to the previous checkpoint. With logging, only a failed process replays its logs, reducing performance overheads and energy consumption~\cite{Elnozahy:2002:SRP:568522.568525}.

%\vspace{-2.2em}
\subsubsection{RMA vs. MP: Uncoordinated Checkpointing}
\label{sec:rma_vs_mp_ucc}

\goal{+ Explain why UC differs in MP and RMA}

%%%%%%
%%%%%%
%%%%%%
%%%%%%
%%%%%%
%%%%%%%In this section we identify and discuss fundamental differences in fault tolerance for message passing and RMA.
%%%%%%
%%%%%%\subsection{Why Resilience for RMA differs from MP? \htor{grammar}}
%%%%%%\label{sec:mp_vs_rma_ft}
%%%%%%
%%%%%%\goal{+ Refer to the title and describe message logging}
%%%%%%
%%%%%%\htor{this should really be uncoordinated CR ...}
%%%%%%To answer this question we start with analyzing message logging in MP
%%%%%%and the analogous \emph{access} (\emph{put} and \emph{get}) logging in RMA.
%%%%%%%

The first and obvious difference is that we now log not \emph{messages}
but \emph{accesses}. Other differences are as follows:

\textbf{Storing Access Logs}
In MP, processes exchange messages
that \emph{always} flow \emph{from} the sender (process $p$) \emph{to} the
receiver (process $q$). Messages can be recorded at the sender's side~\cite{Riesen:2012:ASI:2388996.2389021,Elnozahy:2002:SRP:568522.568525}. During a
recovery, the restored process interacts with other processes to get and reply
the logged messages (see Figure~\ref{fig:mp_rma_simple} (part (1)).

%\vspace{-0.5em}
\begin{figure}[h!]
\centering
\includegraphics[width=0.48\textwidth]{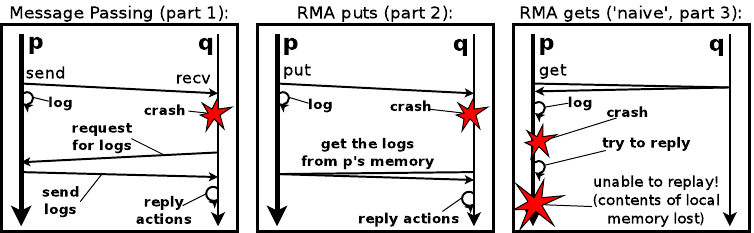}
%\vspace{-1.5em}
\caption{The logging of messages vs. RMA {puts} and {gets} (\cref{sec:rma_vs_mp_cc}).}
\label{fig:mp_rma_simple}
%\vspace{-0.5em}
\end{figure}

\goal{+ Describe logging puts/gets and why it differs from ML}

In RMA, a \textsc{put}\Dptoq\ changes the state of $q$, but a \textsc{get}\Dpfromq\
modifies the state of $p$. Thus, \textsc{put}\Dptoq\ can be logged in $p$'s
memory, but \textsc{get}\Dpfromq\ cannot because a failure of $p$ would prevent a successful
recovery (see Figure~\ref{fig:mp_rma_simple}, part 2 and 3).% We identify more such differences:

\textbf{Transparency of Schemes } 
In MP, both $p$ and $q$ actively participate in communication. In RMA,
$q$ is oblivious to accesses to its
memory and thus any recovery or logging
performed by $p$ can be \emph{transparent} to (i.e., does not obstruct) $q$ (which is usually \emph{not} the case
in MP, cf.~\cite{Riesen:2012:ASI:2388996.2389021}).
%In MP, both sides actively participate in communication. RMA is
%\emph{one-sided}: $q$ is oblivious to accesses to its
%memory. It enables \emph{transparent} logging and recovery: if $p$ logs an access or
%is recovering, no other
%process has to actively participate in it (which is usually \emph{not} the case
%in MP, cf.~\cite{Riesen:2012:ASI:2388996.2389021}).
%
%In MP, both sides actively participate in communication (target explicitly
%receives a message). RMA is \emph{one-sided}: $q$ is completely oblivious to
%any accesses touching its memory. It enables \emph{transparent} logging: if
%$p$ logs an access, no other process has to actively participate in it (which
%is usually \emph{not} the case in MP~\cite{Riesen:2012:ASI:2388996.2389021}).
%\vspace{-0.5em}

\goal{+ Explain piggybacking and why it can't be used in RMA}

\textbf{No Piggybacking } 
Adding some protocol-specific data to messages (e.g., \emph{piggybacking})
is a popular concept in MP~\cite{Elnozahy:2002:SRP:568522.568525}. Still, it cannot be used in RMA because
\textsc{put}s and \textsc{get}s are {accesses}, not {messages}. Yet,
issuing additional accesses is cheap in RMA.

\goal{+ Compare and explain send determinism and access determinism}

\textbf{Access Determinism }
Recent works in MP (e.g.,~\cite{6012907}) explore \emph{send determinism}: the output of an application run is oblivious to the order of
received messages. In our work we identify a similar concept in RMA that
we call \emph{access determinism}. For example, in race-free MPI-3
programs the application output does not depend on the order in which
two accesses $a$ and $b$ committed to memory if $a\ ||_{co}\ b$.

\textbf{Orphan Processes } 
%
%We now discuss more extensively \emph{orphans}~\cite{Alvisi:1998:MLP:630821.631222}.
In some MP schemes (called \emph{optimistic}), senders postpone logging messages for
performance reasons~\cite{Elnozahy:2002:SRP:568522.568525}. 
Assume $q$ received a message $m$ from $p$ and then sent a message $m'$ to $r$. If $q$ crashes
and $m$ is not logged by $p$ at that time, then $q$ may follow a run in
that it
\emph{does not} send $m'$. Thus, $r$ becomes an \emph{orphan}: its state
depends on a message $m'$ that was \emph{not} sent~\cite{Elnozahy:2002:SRP:568522.568525} (see Figure~\ref{fig:orphan_locks}, part 1).

\goal{+ Describe orphans in RMA}

In RMA, a process may also become an orphan. Consider
Figure~\ref{fig:orphan_locks} (part 2). First, $p$ modifies a variable $x$
at $q$. Then, $q$ reads $x$ and conditionally issues a \textsc{put}\Dqtor. If $q$ crashes and $p$ postponed logging \textsc{put}\Dptoq\@, then $q$ (while recovering) may follow a run in which it does not issue \textsc{put}\Dqtor; thus $r$ becomes an orphan.

%Formally, Figure~\ref{fig:orphan_locks} illustrates a run:
%%
%$\textsc{lock}$\ptoq\ $ \xrightarrow{po} \textsc{put}$(\Dptoqn\@, $x = 1$)\ $\xrightarrow{po} \textsc{unlock}$\ptoq\ $\xrightarrow{cohb} \textsc{lock}$\qtoq\ $\xrightarrow{po} ... \xrightarrow{po} \textsc{read}(y = x) \xrightarrow{po} ... \xrightarrow{po} \textsc{unlock}$\qtoq\ $\xrightarrow{po} \textsc{lock}$\qtor\ $ \xrightarrow{po} \textsc{put}$\Dqtor\ $\xrightarrow{po} \textsc{unlock}$\qtor\@.

%\vspace{-0.5em}
\begin{figure}[h!] \centering
\includegraphics[width=0.44\textwidth]{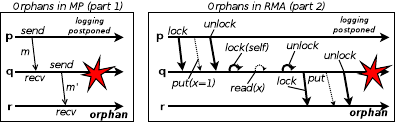}
\vspace{-0.5em}
\caption{Illustration of orphans in MP and RMA (\cref{sec:rma_vs_mp_cc}).} 
\label{fig:orphan_locks}
\vspace{-1.0em}
\end{figure}

%Orphans can be prevented by logging \textsc{put}s \emph{before} closing a matching epoch. Formally, there is no orphan in $\mathcal{D}$ iff:
%
%\small
%\begin{alignat}{2} 
%&\forall_{a \in \mathcal{A}_{comm},\ a.type = put}\ \exists_{s \in \mathcal{A}_{cons}}:\nonumber\\
%&log(a) \xrightarrow{po} s \wedge s.EA = a.EA \wedge s.src = a.src \wedge s.trg = a.trg
%\end{alignat}
%\normalsize
%
%\begin{proof}
%\maciej{See if we need that}
%%Due to space constraints we provide a proof only for locks, the full version can be found in the extended techreport version of this paper\maciej{fine?}\footnote{Released after submission}.
%%Assume by contradiction that the above condition is fulfilled and there is a run in which an arbitrary process $r$ becomes an orphan. Thus, there is a run.. \maciej{trivial?}
%%
%\end{proof}
%
%As we will show later (see Section~\ref{sec:}), this condition has to be strengthened to enable consistent recovery.

%\vspace{-0.5em}
\begin{figure*}
%\vspace{-1.0em}
\centering
\includegraphics[width=1.0\textwidth]{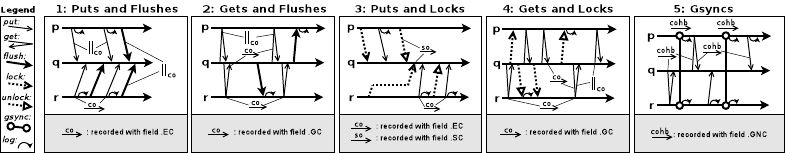}
%\vspace{-1.5em}
\caption{Logging orders $\xrightarrow{so}$, $\xrightarrow{co}$, and $\xrightarrow{hb}$ (\cref{sec:logging_order_info}). In each figure we illustrate example orderings.}
\label{fig:logging_orders_1}
%\vspace{-1.0em}
\end{figure*}

%\vspace{-0.5em}
\subsubsection{Taking an Uncoordinated Checkpoint}
\label{sec:taking_uncoordinated_ckp}

\goal{+ Describe why and how we do uncoordinated checkpoints}

%%We use uncoordinated in-memory checkpoints to trim the logs of \textsc{put}s and
%%\textsc{get}s. 
%%%\htor{sounds odd, uCR is really used to limit the amount of
%%%recomputation}
We denote the $i$th uncoordinated checkpoint taken by process $p$ as
$C_{p}^{i}$. 
%We model the $i$th uncoordinated checkpoint taken by process $p$ as
%$\{C_{p}^{i}\} \subseteq \mathcal{I}$. 
%%
%To keep the notation simple we denote $\{C_{p}^{i}\}$ with its checkpoint action $C_{p}^{i}$.
%
%%%%We denote the $i$th uncoordinated checkpoint taken by process $p$ as
%%%%$C_{p}^{i}$. 
%%%%%
%%%%To keep the notation simple we identify $C_{p}^{i}$ with its checkpoint action, i.e., $C_{p}^{i} \equiv \{C_{p}^{i}\} \subseteq \mathcal{I}$.
%%%%%
Taking $C_{p}^{i}$ is simple and entails: (1) locking local application data,
(2) sending the copy of the data to some remote volatile storage, and (3) unlocking the application data
(we defer the discussion on the implementation details until Section~\ref{sec:libImplementation}).
After $p$ takes $C_{p}^{i}$, any process $q$ can delete the logs of every
\textsc{put}\Dqtop\ (from $LP_q[p]$) and \textsc{get}\Dpfromq\ (from $LG_q[p]$) 
that committed in $p$'s memory before $C_{p}^{i}$
(i.e., \textsc{put}\Dqtop\ $\xrightarrow{co} C_{p}^{i}$ and 
\textsc{get}\Dpfromq\ $\xrightarrow{co} C_{p}^{i}$).

\goal{+ Describe a correctness condition for uncoordinated checkpoints}

We demand that every $C_{p}^{i}$ is taken \emph{immediately after}
closing/opening an epoch and \emph{before} issuing any new communication
operations (we call this the \emph{epoch condition}).
This condition is required because, if $p$ issues a \textsc{get}\Dpfromq\@, 
the application data is guaranteed to be consistent only after closing the epoch.

%In RMA, due to its richer semantics, there is also an \emph{epoch} condition: any $C_{p}^{i}$ has to be taken \emph{after} closing all the epochs and \emph{before}
%issuing any new communication calls. This is because if $p$ issues a \textsc{get}\Dpfromq\@, 
%the application data is guaranteed to be consistent only after closing the epoch.

%\vspace{1.5em}
\subsubsection{Transparent Logging of RMA Accesses}
\label{sec:transparentLogging}

%\htor{this falls out of nowhere ... you shouls start with
%coordinated/uncoordinated CR etc.}

%%%%%

\goal{+ Say why we want logging in RMA}

%We now proceed to describe the logging of RMA memory accesses. For both
%\textsc{put}s and \textsc{get}s we log the arguments (target rank,
%displacement, etc.), the carried data, and the active epoch number $EA$\ptoq\@.
%In the following table we present data structures used in the section.

%We now describe diskless logging of RMA memory accesses. 

%%%%Logging of accesses in RMA brings benefits similar to recording messages
%%%%in MP. For example, if no logging is enabled, a single process crash forces all other processes 
%%%%to rollback to the previous checkpoint. With logging, only a failed process replays its logs, reducing performance overheads and energy consumption~\cite{Elnozahy:2002:SRP:568522.568525}.
%
%%%In the access logging scheme both
We now describe the logging of \textsc{put}s and \textsc{get}s; all necessary data structures are shown in Table~\ref{tab:str}.
%\vspace{1.0em}

\begin{table}
{\centering
%\vspace{0.5em}
\small
%\begin{tabular}{>{\centering\arraybackslash}m{1cm} p{7.2cm}}
\begin{tabular}{>{\centering\arraybackslash}m{1.5cm} p{6.3cm}} \toprule
\textbf{Structure}&\textbf{Description}\\ \midrule

$LP_{p}[q] \in \mathcal{S}$&\makecell[l]{Logs of \textsc{put}s issued by $p$ and targeted at
$q$.}\\ \midrule

$LG_{q}[p] \in \mathcal{S}$&\makecell[l]{Logs of \textsc{get}s targeted at $q$ and issued by
$p$.}\\ \midrule

$LP_{p} \in \mathcal{S}$&\makecell[l]{Logs of \textsc{put}s issued and stored by $p$ and targeted\\at any
other process; $LP_{p} \equiv \bigcup_{r \in \mathcal{P} \wedge r \neq p}
LP_{p}[r]$.}\\ \midrule

$LG_{q} \in \mathcal{S}$&\makecell[l]{Logs of $gets$ targeted and stored at $q$, issued by\\any
other process; $LG_{q} \equiv \bigcup_{r \in \mathcal{P} \wedge r \neq q}
LG_{q}[r]$.}\\
\midrule

$Q_p \in \mathcal{S}$&\makecell[l]{A helper container stored at $p$, used to\\temporarily log \#\textsc{get}s issued by $p$.}\\
\midrule

$N_q[p] \in \mathcal{S}$&\makecell[l]{A structure (stored at $q$) that determines\\whether or not $p$ issued a non-blocking\\ \textsc{get}\Dptoq\ ($N_q[p] = true$ or $false$, respectively)}\\
\bottomrule

%$Q_p \in \mathcal{S}$&\makecell[l]{A helper container that we use to store the operations\\(issued by
%$p$) to be recorded during the next\\memory consistency call.}\\ \bottomrule

\end{tabular} \normalsize}
%\vspace{-0.5em}

%Every structure can be modified with two operations: $.add()$ (insert an
%element) and $.del()$ (remove and return an element).  

%We also use a function
%$\mathscr{S}: \mathcal{S} \to \mathcal{D}$ that determines the storage place
%for every structure.

\caption{Data structures used in RMA logging (\cref{sec:transparentLogging}). $LP_{p}[q]$ and $LP_p$ are stored at $p$. $LG_{q}[p]$ and $LG_q$ are stored at $q$.}
\label{tab:str}
%\vspace{-1.0em}
\end{table}

\goal{++ Describe logging of puts}

%\vspace{-2.2em}
%%\subsubsection{Logging Puts} \label{sec:loggingPuts}
\textbf{Logging Puts } 
%
%
%
%
%The logging of a \textsc{put}\Dptoq\ is illustrated in
%Algorithm~\ref{alg:log_puts}.
%Self-locking ($LP_{p}$) is necessary because, even if it is only $p$ that can
%modify the contents of $LP_{p}[q]$, there may be other processes being
%recovered which may try to read $LP_p$. Atomicity between logging and putting is not required because, in the weak
%consistency memory model, the source memory of the put operation may not be
%modified until the current epoch ends. If the program modifies it nevertheless,
%RMA implementations are allowed to return any value, thus, the logged value is
%irrelevant\maciej{formalize?}. We log \textsc{put}\Dptoq\ before closing epoch \textsc{put}\Dptoq$.EA$.
%
To log a \textsc{put}\Dptoq\@, $p$ first calls \textsc{lock}$(p \rightarrow p, LP_{p})$.
Self-locking is necessary because there may be other processes being
recovered that may try to read $LP_p$. Then, the \textsc{put} is logged ($LP_{p}[q] := LP_{p}[q]\ \cup\ \{$\textsc{put}\Dptoq\};  ``:='' denotes the assignment of
a new value to a variable or a structure).
Finally, $p$ unlocks $LP_p$.
Atomicity between logging and putting is not required because, in the weak
consistency memory model, the source memory of the put operation may not be
modified until the current epoch ends. If the program modifies it nevertheless,
RMA implementations are allowed to return any value, thus the logged value is
irrelevant. We log \textsc{put}\Dptoq\ before closing
the epoch \textsc{put}\Dptoq$.EC$. If the \textsc{put} is blocking then we 
log it before issuing, analogously to the \emph{pessimistic} message logging~\cite{Elnozahy:2002:SRP:568522.568525}.

%\htor{what about blocking UPC puts?}

%\begin{algorithm}
%%\setstretch{1.0}
%\DontPrintSemicolon
%\KwIn{$put := \textsc{put}(p \mathrel{\substack{\textstyle\Rightarrow\\[-0.5ex]
%                      \textstyle\rightarrow}} q)$}
%\textsc{lock}$(p \rightarrow p, LP_{p})$\;
%$LP_{p}[q] := LP_{p}[q] \cup put$\;
%\textsc{unlock}$(p \rightarrow p, LP_{p})$\;
%\caption{Logging \textsc{put}s}
%\label{alg:log_puts}
%\end{algorithm}

\goal{++ Describe logging of gets}

%%%%\subsubsection{Logging Gets} \label{sec:loggingGets}

\textbf{Logging Gets }
%
% is logged in two phases: 
%logging the determinant $\#\textsc{get}$ (phase 1, before closing the epoch) and logging the actual $\textsc{get}$ (phase 2, after closing the epoch).
%
We log a \textsc{get}\Dpfromq\ in two phases to retain its asynchronous behavior (see Algorithm~\ref{alg:log_gets}). First, we record the determinant \#\textsc{get}\Dpfromq\ in
$Q_p$ (lines 2-3). We cannot access \textsc{get}\Dpfromq$.data$ as the
local memory will only be valid after the epoch ends. We avoid issuing an
additional blocking \textsc{flush}\ptoq\@, instead we rely on the user's
call to end the epoch.
Second, when the user ends the epoch, we lock the remote log $LG_{q}$, record \textsc{get}\Dpfromq\@, and
unlock $LG_q$ (lines 4-7).
%using several remote memory accesses

%and increment appropriately the epoch number.

%To protect the system against failures that could happen between issuing
%$gets$ and closing the epoch, we utilize a special local marker $MRK$\ptoq\ ($M$a$RK$er). We set this marker to true after issuing the first
%$get$\ptoq\ and to false after closing the logging epoch. During the recovery,
%if a recovered process notices that 

\goal{++ Describe how me manage faults happening before the epoch ends}

Note that if $p$ fails between issuing
\textsc{get}\Dpfromq\ and closing the epoch, it will not be able to replay it
consistently. To address this problem, $p$ sets $N_q[p]$ at process $q$ to \emph{true} right before issuing the first
\textsc{get}\Dptoq\ (line 1), and to \emph{false} after closing the epoch \textsc{get}\Dptoq$.EC$ (line 8). During the recovery,
if $p$ notices that any $N_q[p] = true$, it falls back to another resilience mechanism (i.e., the last
coordinated checkpoint). If the \textsc{get} is blocking then we set $N_q[p] = false$ after returning from the call.

%\vspace{-0.5em}
\begin{algorithm}
%\setstretch{1.0}
\scriptsize
\DontPrintSemicolon
\KwIn{$get := \textsc{get}(p \mathrel{\substack{\textstyle\Leftarrow\\[-0.5ex]
                      \textstyle\rightarrow}} q)$}
                     
\tcc{\scriptsize Phase 1: starts right before issuing the $get$}
$N_q[p] := true$\;
\tcc{\scriptsize Now we issue the $get$ and log the $\#get$}
issue $\textsc{get}(p \mathrel{\substack{\textstyle\Leftarrow\\[-0.5ex]\textstyle\rightarrow}} q)$\;
$Q_{p} \gets Q_{p} \cup \#get$\;   

\tcc{\scriptsize Phase 2: begins after ending the epoch $get.EC$}
\textsc{lock}$(p \rightarrow q, LG_{q})$\;
$LG_{q}[p] := LG_{q}[p] \cup get$\;
$Q_{p} := Q_{p}\ \textbackslash\ \#get$\;   
\textsc{unlock}$(p \rightarrow q, LG_{q})$\;
$N_q[p] := false$\;
%$EA(p \to q) := EA(p \to q) + 1$\;
%$EP(p \to q) := EP(p \to q) + 1$\;
%set($MRK(p \to q),0$)\;
%issue-put(args,payload)\;
\caption{Logging \emph{gets} (\cref{sec:transparentLogging})}
%\vspace{-0.5em}
\label{alg:log_gets}
\end{algorithm}

\vspace{-1em}
\section{Causal Recovery for UC}
\label{sec:complicated}

%We now show how to causally recover a failed process.
%To do it, we provide an in-depth analysis of RMA synchronization and
%consistency schemes. 
%

\goal{Describe in general a causal recovery and summarize the section}

We now show how to causally recover a failed process (\emph{causally} means preserving
$\xrightarrow{co}$, $\xrightarrow{so}$, and $\xrightarrow{hb}$).
This section describes technical details on how to guarantee all orders to
ensure a correct access replay. If the reader is not interested in all
details, she may proceed to Section~\ref{sec:divisionIntoGroups} without
disrupting the flow.
A {causal} process recovery has three phases: (1) fetching uncoordinated checkpoint data, 
(2) replaying accesses from remote logs, and (3) in case of a problem during the replay, falling
back to the last coordinated checkpoint.
We first show how we log the respective orderings between accesses (Section~\ref{sec:logging_order_info}) and
how we prevent replaying some accesses twice (Section~\ref{sec:managing_unc}). We finish with
our recovery scheme (Section~\ref{sec:recovery_scheme}) and a discussion (Section~\ref{sec:discussion_rec}).
%
%Due to space constraints, we include full
%proofs in the techreport version of the paper\footnote{\scriptsize
%http://spcl.inf.ethz.ch/Research/Parallel\_Programming/ftRMA}.
%The correctness proof of the scheme shows that (1) it replays all recorded logs or appropriately falls back to the coordinated scheme, (2) it replays accesses causally, and (3) it does not deadlock.
%

%\subsubsection{Synchronization and Consistency Calls}
%\label{sec:synchConsCalls}

\subsection{Logging Order Information}
\label{sec:logging_order_info}

\goal{+ Describe our explanation methodology in this section}

%\emph{Causal} recovery has to replay accesses in order. In our model,
%the orders to be preserved are $\xrightarrow{so}$ (and thus $\xrightarrow{hb}$), and $\xrightarrow{co}$. 
%We first show how our logging schemes record such
%ordering information. 
We now show how to record
$\xrightarrow{so}$, $\xrightarrow{hb}$, and $\xrightarrow{co}$.
For clarity, but without loss of generality, we
separately present several scenarios that exhaust possible
communication/synchronization patterns in our model. We
consider three processes ($p$, $q$, $r$) and we 
analyze what data is required to replay $q$. We show each
pattern in Figure~\ref{fig:logging_orders_1}.

\goal{+ Describe logging of consistency order when using puts/flushes}

\textbf{A. Puts and Flushes }
First, $p$ and $r$ issue \textsc{put}s and \textsc{flush}es 
at $q$. At both $p$ and $r$, \textsc{put}s separated by \textsc{flush}es are ordered 
with $\xrightarrow{co}$. This order is preserved by recording epoch counters ($.EC$) with
every logged \textsc{put}\Dptoq\@. Note that, however, RMA semantics
\emph{do not} order calls issued by $p$ and $r$: $[\textsc{put}$\Dptoq\
$||_{co}\ \textsc{put}$\Drtoq] without additional process
synchronization. 
Here, we assume \emph{access determinism}: the recovery output does not depend on the order in which such \textsc{put}s committed in $q$'s memory.
%, and $\textsc{flush}$\Dptoq\ $||_{co}\ \textsc{flush}$\Drtoq].

%Finally, $\xrightarrow{po}$ between
%$put$s in the same epoch 

\goal{+ Describe logging of consistency order when using gets/flushes}

\textbf{B. Gets and Flushes }
Next, $q$ issues \textsc{get}s and \textsc{flush}es 
targeted at $p$ and $r$. Again, $\xrightarrow{co}$ has to be logged.
However, this time \textsc{get}s targeted at \emph{different}
processes \emph{are} ordered (because they are issued by the same process). To log this ordering, $q$ maintains a local \emph{Get Counter}
$GC_q$ that is incremented each time $q$ issues a \textsc{flush}$(q \to
\diamond)$ to any other process.
The value of this counter is logged with each \textsc{get} using the field $.GC$ (cf. Section~\ref{sec:formal_model}).

\goal{+ Describe logging of synchronization order when using puts/locks}

\textbf{C. Puts and Locks }
In this scenario $p$ and $r$ issue \textsc{put}s at $q$ and synchronize their accesses
with \textsc{lock}s and \textsc{unlock}s. This pattern requires logging the $\xrightarrow{so}$ order.
We achieve this with a \emph{Synchronization Counter} $SC_q$ stored at $q$. After issuing a \textsc{lock}\ptoq\@, $p$ (the same refers to $r$) fetches the value of $SC_q$, increments it, updates remote $SC_q$, and records it with every \textsc{put} using the field
$.SC$ (cf. Section~\ref{sec:formal_model}). In addition, this scenario requires recording $\xrightarrow{co}$ that we solve with $.EC$, analogously as in the ``Puts and Flushes'' pattern.

\goal{+ Describe logging of synchronization order when using gets/locks}

\textbf{D. Gets and Locks }
Next, $q$ issues \textsc{get}s and uses \textsc{lock}s 
targeted at $p$ and $r$. This pattern is solved analogously to
the ``Gets and Flushes'' pattern.

\goal{+ Describe logging of consistency and HB order when using gsyncs}

\textbf{E. Gsyncs }
%We assume that a \textsc{gsync} may also introduce $\xrightarrow{hb}$.
% and simplest
The final pattern are \textsc{gsync}s (that may again introduce $\xrightarrow{hb}$) combined with any communication
action. Upon a \textsc{gsync}, each process $q$ increments its
\emph{GsyNc Counter} $GNC_q$ that is logged in an actions' $.GNC$ field
(cf. Section~\ref{sec:formal_model}).% to maintain $\xrightarrow{cohb}$.

\begin{algorithm}[h!]
%\begin{algorithmic}[1]
\scriptsize
%\small
%\Indm
%\setstretch{1.0}
\SetAlgoLined\DontPrintSemicolon
\SetKwFunction{recovery}{recovery}
\SetKwFunction{logsWithMinCnt}{logsWithMinCnt}
\SetKwFunction{replayEachAction}{replayEachAction}
\SetKwFunction{fetchCheckpointData}{fetchCheckpointData}
\SetKwProg{myalg}{Function}{}{}

\myalg{\recovery{}}{
	fetch\_checkpoint\_data()\;
	put\_logs := \{\}; get\_logs := \{\}\;
	%all\_logs := \{\}\;
	\ForAll{$q \in \mathcal{P}:\ q \neq p_{new}$}{
%		\textsc{lock}$(p_{new} \rightarrow q, LP_q[p_f])$; \textsc{lock}$(p_{new} \rightarrow q, LG_q[p_f])$\;
		\textsc{lock}$(p_{new} \rightarrow q)$\;
		\If{$N_q[p_f] = 1 \lor M_q[p_f] = true$}{
			\tcc{\scriptsize Stop the recovery and fall back to the last coordinated checkpoint}
		}
		put\_logs := put\_logs\  $\cup LP_q[p_f]$\;
		get\_logs := get\_logs\  $\cup LG_q[p_f]$\;
		\textsc{unlock}$(p_{new} \rightarrow q)$\;
		%all\_logs := all\_logs $\cup LP_q[p_f] \cup LG_q[p_f]$\;
%		\textsc{unlock}$(p_{new} \rightarrow q, LG_q[p_f])$; \textsc{unlock}$(p_{new} \rightarrow q, LP_q[p_f])$\;
	}
%	\textsc{lock}$(p_{new} \rightarrow p_{new})$\;

	%$\bigcup_{p \in \mathcal{P}: p \neq p_{new}}\ LP_p \cup LG_p$\;
	\While{|put\_logs| > 0 $\lor$ |get\_logs| > 0}{
		%\tcc{Return the set with logs that have the smallest value of $.GNC$ counter}
%		gnc\_logs := logsWithMinCnt(GNC, put\_logs) $\cup$ logsWithMinCnt(GNC, get\_logs)\;
		gnc\_logs := logsWithMinCnt(GNC, put\_logs $\cup$ get\_logs)\;
	
		\While{|gnc\_logs| > 0}{
			%\tcc{Return the set with logs that have the smallest value of $.SC$ counter}
			gnc\_put\_logs := gnc\_logs $\cap$ put\_logs\;
			gnc\_get\_logs := gnc\_logs $\cap$ get\_logs\;
			ec\_logs := logsWithMinCnt(EC, gnc\_put\_logs)\;
			gc\_logs := logsWithMinCnt(GC, gnc\_get\_logs)\;
			%\tcc{Return the set with logs that have the smallest value of $.GC$ counter}
			replayEachAction(ec\_logs)\;
			replayEachAction(gc\_logs)\;
			gnc\_logs := gnc\_logs \textbackslash\ (ec\_logs $\cup$ gc\_logs)\;
		}
		put\_logs := put\_logs \textbackslash\ gnc\_logs\; 
		get\_logs := get\_logs \textbackslash\ gnc\_logs\; 
	}
%	\textsc{unlock}$(p_{new} \rightarrow p_{new})$\;
	\KwRet\;
}
%setcounter{AlgoLine}{0}
\SetKwProg{myproc}{Function}{}{}

\myalg{\logsWithMinCnt{Counter, Logs}}{
\scriptsize
\tcc{Return a set with logs from $Logs$ that have the smallest value of the specified counter (one of:$GNC, EC, GC, SC$).}
}

\myalg{\replayEachAction{Logs}}{
\scriptsize
\tcc{Reply each log from set $Logs$ in any order.}
}

\myalg{\fetchCheckpointData{}}{
\scriptsize
\tcc{Fetch the last checkpoint and load into the memory.}
}

%\myproc{\proc{}}{
%  \nl xxx\;
%  \nl \KwRet\;}
%
%\Function{getLogsWithMinCounter}{Counter, Logs} 
%\tcc{Return a set with logs from $Logs$ that have the smallest value of the specified counter (one of: $GNC, EC, GC, SC$)}
%\EndFunction
%
%\Function{replayEachAction}{Counter, Logs} 
%\EndFunction

%set($MRK(p \to q),0$)\;
%issue-put(args,payload)\;
%\end{algorithmic}
\caption{The causal recovery scheme for codes that synchronize with \textsc{gsync}s (\cref{sec:recovery_scheme}, \cref{sec:discussion_rec}).}
\label{alg:recovery}
\end{algorithm}

\subsection{Preventing Replaying Accesses Twice}
\label{sec:managing_unc}

\goal{+ Describe the problem of puts that modify memory twice}

Assume that process $p$ issues a \textsc{put}\Dptoq\ (immediately logged
by $p$ in $LP_p[q]$) such that \textsc{put}\Dptoq\ $\xrightarrow{co}
C_{q}^{j}$. It means that the state of $q$ recorded in checkpoint
$C_{q}^{j}$ is affected by \textsc{put}\Dptoq\@. Now assume that $q$
fails and begins to replay the logs. If $p$ did not delete the log of \textsc{put}\Dptoq\ from $LP_p[q]$ (it was allowed to do it after $q$ took $C_{q}^{j}$), then $q$ replays \textsc{put}\Dptoq\ and this \textsc{put} affects its memory \emph{for the second time}. This is not a problem if \textsc{put}\Dptoq$.combine = false$, because such a \textsc{put}
always overwrites the memory region with the same value. However, if \textsc{put}\Dptoq$.combine = true$, then $q$ ends up
in an inconsistent state (e.g., if this \textsc{put} increments a memory cell, this cell will be incremented twice).

\goal{+ Describe the solution to the above problem}

To solve this problem, every process $p$ maintains a local structure $M_p[q] \in \mathcal{S}$. When $p$ issues and logs a \textsc{put}\Dptoq\ such that \textsc{put}\Dptoq$.combine = true$, it sets $M_p[q] := true$. When $p$ deletes \textsc{put}\Dptoq\ from its logs, it sets $M_p[q] := false$. If $q$ fails, starts to recover, and sees that any $M_p[q] = true$, it stops the recovery and falls back to the coordinated checkpoint. This scheme is valid if access determinism is assumed. Otherwise we set $M_p[q] := true$ regardless of the value of \textsc{put}\Dptoq$.combine$; we use the same approach if $q$ can issue \textsc{write}s to the memory regions accessed with remote \textsc{put}s parallel in $||_{co}$ to these \textsc{write}s.% parallel to \textsc{write}s issued by $q$.

\vspace{+2.2em}
\subsection{Recovering a Failed Process}
\label{sec:recovery_scheme}

\goal{+ Describe the first part of the recovery}

We now describe a protocol for codes that synchronize with \textsc{gsync}s.
Let us denote the failed process as $p_{f}$. We assume an underlying batch system that provides a new process $p_{new}$ in the place of $p_{f}$, and that other processes resume their communication with $p_{new}$ after it fully recovers. We illustrate the scheme in Algorithm~\ref{alg:recovery}. First, $p_{new}$ fetches the checkpointed data.
%Then, $p_{new}$ fetches the checkpointed data (line 3; we again skip the implementation details of this mechanism and discuss them in~Section~\ref{sec:libImplementation}).
%
Second, $p_{new}$ gets the logs of \textsc{put}s (put\_logs) and
\textsc{get}s (get\_logs) related to $p_f$ (lines 3-11). It also checks
if any $N_q[p_f] = true$ (see~\cref{sec:transparentLogging}) or
$M_q[p_f] = true$ (see~\cref{sec:managing_unc}), if yes it instructs all
processes to roll back to the last coordinated checkpoint. The protocol uses \textsc{lock}s (lines 5,10) to prevent data races due to, e.g., concurrent recoveries and log cleanups by $q$.% (see\footnotemark[1] for details).

\goal{+ Describe the main part of the recovery}

The main part (lines 12-27) replays accesses causally. The recovery ends when there are no logs left (line 12; $|logs|$ is the size of the set ``logs''). We first get the logs with the smallest $.GNC$ (line 13) to maintain $\xrightarrow{cohb}$ introduced by \textsc{gsync}s (see~\cref{sec:logging_order_info} E). Then, within this step, we find the logs with minimum $.EC$ and $.GC$ to preserve $\xrightarrow{co}$ in issued \textsc{put}s and \textsc{get}s, respectively (lines 18-19, see~\cref{sec:logging_order_info} A, B). We replay them in lines 20-21. 
%%%Independently of $\xrightarrow{co}$ and $\xrightarrow{so}$, but within the same gsync step, we replay logs of \textsc{get}s (bc\_lg\_logs) that follow $\xrightarrow{co}$ determined by $.GC$ (see~\cref{sec:logging_order_info} B, D). After all logs are replayed, $p$ unlocks itself (line 31).

\begin{theorem}
The recovery scheme presented in Algorithm~\ref{alg:recovery} replays each fetched action exactly once.
\end{theorem}

% and $\neg( \exists b \in $ put\_logs $\cup$ get\_logs $: b.GNC = a.GNC)$

\begin{proof}
Consider the gnc\_logs set obtained in line 13. The definition of function logsWithMinCnt ensures that, after executing the action in line 13 and before entering the loop that starts in line 15, every $a \in $ gnc\_logs has identical $a.GNC$. Then, 
the condition in line 15 together with the actions in line 22 and the definition of logsWithMinCnt ensure that gnc\_logs is empty when the loop in lines 15-23 exits (all logs in gnc\_logs are replayed). This result, together with the actions in lines 24-25, guarantee that each action $a$ obtained in the lines 4-11 is extracted from log\_puts $\cup$ log\_gets in line 13 and replayed exactly once.
\end{proof}

\begin{theorem}
The recovery scheme presented in Algorithm~\ref{alg:recovery} preserves the $\xrightarrow{cohb}$ order introduced by \textsc{gsync}s (referred to as the \emph{gsync order}).
\end{theorem}

\begin{proof}
Let us denote the action that replays communication action $a$ at process $p_{new}$ as $\mathcal{R}(a,p_{new}) \in \mathcal{I}$ (as $\mathcal{R}(a,p_{new})$ affects only the memory of the calling process $p_{new}$, it is an internal action). 
%
%We refer to the $\xrightarrow{cohb}$ order introduced by \textsc{gsync}s as the \emph{gsync order}. 
%
Assume by contradiction that the gsync order is not preserved while recovering. Thus, $\exists a_1, a_2 \in \text{log\_puts} \cup \text{log\_gets}:\ (a_1.GNC > a_2.GNC) \wedge (\mathcal{R}(a_1,p_{new}) \xrightarrow{po} \mathcal{R}(a_2,p_{new}))$.
It means that the action of including $a_1$ into gnc\_logs (line 13) took place before the analogous action for $a_2$ (in the $\xrightarrow{po}$ order). But this contradicts the definition of function logsWithMinCnt($GNC$, set) that returns all the actions from set that have the minimum value of the $GNC$ counter.
\end{proof}

\begin{algorithm}[h!]
%\begin{algorithmic}[1]
\scriptsize
%\small
%\Indm
%\setstretch{1.0}
\SetAlgoLined\DontPrintSemicolon
\SetKwFunction{recovery}{recovery}
\SetKwFunction{logsWithMinCnt}{logsWithMinCnt}
\SetKwFunction{replayEachAction}{replayEachAction}
\SetKwFunction{fetchCheckpointData}{fetchCheckpointData}
\SetKwProg{myalg}{Function}{}{}

\myalg{\recovery{}}{
	fetch\_checkpoint\_data()\;
	put\_logs := \{\}\;
	%all\_logs := \{\}\;
	\ForAll{$q \in \mathcal{P}:\ q \neq p_{new}$}{
%		\textsc{lock}$(p_{new} \rightarrow q, LP_q[p_f])$; \textsc{lock}$(p_{new} \rightarrow q, LG_q[p_f])$\;
		\textsc{lock}$(p_{new} \rightarrow q)$\;
		\If{$M_q[p_f] = true$}{
			\tcc{\scriptsize Stop the recovery and fall back to the last coordinated checkpoint}
		}
		put\_logs := put\_logs\  $\cup LP_q[p_f]$\;
		\textsc{unlock}$(p_{new} \rightarrow q)$\;
		%all\_logs := all\_logs $\cup LP_q[p_f] \cup LG_q[p_f]$\;
%		\textsc{unlock}$(p_{new} \rightarrow q, LG_q[p_f])$; \textsc{unlock}$(p_{new} \rightarrow q, LP_q[p_f])$\;
	}
%	\textsc{lock}$(p_{new} \rightarrow p_{new})$\;

	%$\bigcup_{p \in \mathcal{P}: p \neq p_{new}}\ LP_p \cup LG_p$\;
	\While{|put\_logs| > 0}{
		%\tcc{Return the set with logs that have the smallest value of $.GNC$ counter}
%		gnc\_logs := logsWithMinCnt(GNC, put\_logs) $\cup$ logsWithMinCnt(GNC, get\_logs)\;
		sc\_put\_logs := logsWithMinCnt(SC, put\_logs)\;
	
		\While{|sc\_put\_logs| > 0}{
			%\tcc{Return the set with logs that have the smallest value of $.SC$ counter}
			ec\_logs := logsWithMinCnt(EC, sc\_put\_logs)\;
			%\tcc{Return the set with logs that have the smallest value of $.GC$ counter}
			replayEachAction(ec\_logs)\;
			sc\_put\_logs := sc\_put\_logs \textbackslash\ ec\_logs\;
		}
		put\_logs := put\_logs \textbackslash\ sc\_put\_logs\; 
	}
%	\textsc{unlock}$(p_{new} \rightarrow p_{new})$\;
	\KwRet\;
}
%setcounter{AlgoLine}{0}
\SetKwProg{myproc}{Function}{}{}

\myalg{\logsWithMinCnt{Counter, Logs}}{
\scriptsize
\tcc{Return a set with logs from $Logs$ that have the smallest value of the specified counter (one of:$GNC, EC, GC, SC$).}
}

\myalg{\replayEachAction{Logs}}{
\scriptsize
\tcc{Reply each log from set $Logs$ in any order.}
}

\myalg{\fetchCheckpointData{}}{
\scriptsize
\tcc{Fetch the last checkpoint and load into the memory.}
}

%\myproc{\proc{}}{
%  \nl xxx\;
%  \nl \KwRet\;}
%
%\Function{getLogsWithMinCounter}{Counter, Logs} 
%\tcc{Return a set with logs from $Logs$ that have the smallest value of the specified counter (one of: $GNC, EC, GC, SC$)}
%\EndFunction
%
%\Function{replayEachAction}{Counter, Logs} 
%\EndFunction

%set($MRK(p \to q),0$)\;
%issue-put(args,payload)\;
%\end{algorithmic}
\caption{The causal recovery scheme for codes that synchronize with \textsc{lock}s and communicate with \textsc{put}s (\cref{sec:recovery_scheme}, \cref{sec:discussion_rec}).}
\label{alg:recovery_locks}
\end{algorithm}

We now present a recovery scheme for codes that synchronize with \textsc{lock}s and communicate with \textsc{put}s.
The first part of the scheme is identical to the one that targets \textsc{gsync}s; the difference is that we do not have to check the values of $N_q[p_f]$.

\goal{+ Describe the main part of the recovery}

In the main part (lines 11-20) actions are replayed causally. We first get the logs with the smallest $.SC$ (line 12) to maintain $\xrightarrow{so}$ introduced by \textsc{lock}s (see~\cref{sec:logging_order_info} C). Then, within this step, we find the logs with minimum $.EC$ to preserve the $\xrightarrow{co}$ order (line 4, see~\cref{sec:logging_order_info} A). The \textsc{put}s are replayed in line 15.

%%%%
%%%%\goal{+ Describe the proof and say it's in TR}
%%%%
%%%%The correctness proof of the scheme shows that (1) it replays all recorded logs or appropriately falls back to the coordinated scheme, (2) it replays accesses causally, and (3) it does not deadlock.
%%%%%
%%%%Due to space constraints, we include full
%%%%proofs in the techreport version of the paper\footnote{\scriptsize http://spcl.inf.ethz.ch/Research/Parallel\_Programming/frma}.

%
%We use Function \texttt{logsWithMinCnt} to get the logs with the minimum value of a specified counter; it enables replaying logs in correct (increasing) order.

%Other processes
%that try to access $p_{new}$ between failure and full
%recovery will simply wait until $p_{new}$ unlocks itself.
%Remote processes are not obstructed by the recovering process,
%as opposed to traditional message logging protocols,
%where all remote processes need to play an active role in
%the recovery, which often causes a significant performance
%penalty.

%\begin{proof}
%Prove that every process replays all required accesses.
%\end{proof}
%
%\begin{proof}
%Prove than replayed sentence is causal.
%\end{proof}
%
%\begin{proof}
%Prove no deadlocks.
%\end{proof}

\subsection{Discussion}
\label{sec:discussion_rec}

\goal{+ Describe the memory-performance tradeoffs}

%\textbf{Performance }
%
Our recovery schemes present a trade-off between memory efficiency and time to recover.
Process $p_{new}$ fetches all related logs and only then begins to replay accesses. Thus,
we assume that its memory has capacity to contain put\_logs and get\_logs; a reasonable
assumption if the user program has regular communication patterns (true for most of today's
RMA applications~\cite{fompi-paper}). A more memory-efficient scheme fetches logs
while recovering. This incurs performance issues as $p_{new}$ has to access remote logs multiple times.
%We present this scheme in the techreport.

%As was discussed in Section \ref{sec:loggingGets}, if any $EP_{q,m}[p_f]$ is equal to 1, then the system will fall back to the last coordinated checkpoint. As the protocol locks $p_{new}$ for the period of recovery, if any other $CM$ process has to communicate with $p_{new}$ during the recovery, it simply blocks until the recovery is completed.

%%\vspace{-0.5em}
%\section{Model Extensions}
%\label{sec:extendModelToday}
%
%We now propose a \emph{demand checkpointing} mechanism and a \emph{failure domain hierarchy} to
%address two challenges of today's petascale and future exascale machines: diminishing amounts of memory per code and deep hardware hierarchies~\cite{}). 

%\vspace{-0.8em}
%%%%%\section{Hardware Hierarchies and Concurrent Crashes}
\section{Extending the Model for more Resilience}
\label{sec:divisionIntoGroups}

\goal{Motivate and explain our model extensions}

%In this section we describe the strategy that increases the resiliency of our three-level scheme. This scheme is orthogonal to the protocol and can be utilized in any other system. Its main goal is to decrease the rate of failures that force the restart of the user computations.

Our model and in-memory resilience schemes
are oblivious to the underlying hardware.
However, virtually all of today's systems have
a hierarchical hardware layout (e.g., cores
reside on a single chip, chips reside in a single node, nodes form a rack, and racks form a cabinet).
Multiple elements may be affected by a single
failure at a higher level, jeopardizing the safety
of our protocols.
%Processes are typically mapped to cores on the nodes. Thus,
%if a node fails (e.g., due to a failing memory DIMM), multiple
%processes will fail.
We now extend our model to
cover arbitrary hierarchies
and propose \emph{topology-aware} mechanisms 
to make our schemes handle concurrent hardware failures.
Specifically, we propose three following extensions:
%that
%improve the resilience of proposed schemes by
%handling multiple concurrent crashes.
% analyze how such layout impacts the resilience of our mechanisms.
% 

%minimizes the probability of nonrecoverable failures
%and utilize the 

%%\vspace{-2.0em}
%\paragraph{RMA-specific concepts}
%
%Besides the issues that are common to RMA and message passing (\emph{catastrophic failures and topology-aware groups}) we also analyze in this section several novel concepts in the field of fault-tolerance for RMA. These concepts are \emph{epoch failures} and the \emph{synchronization of failure domains}.

%\vspace{-1.3em}
%%%%\subsection{Extending the Model for more Resilience}

%We now extend the model of a distributed system $\mathcal{D}$.

\goal{+ Describe failure domain hierarchies and how we model them}

%\subsubsection{The Hierarchy of Failure Domains}
\textbf{The Hierarchy of Failure Domains }
%
%We introduce a hierarchy of \emph{failure domains} (based on the hardware layout of the machine). A \emph{failure domain} is an element of the system that can fail (e.g., a node, a PSU (Power Supply Unit), a chassis, or a cabinet). An example hierarchy is illustrated is Figure~\ref{fig:failureHierarchy}. We denote $h$ as the number of levels of the failure domain hierarchy (in Figure~\ref{fig:failureHierarchy}, $h = 4$).  We skip the level of single cores because in practice the smallest failure domain is a compute node (e.g., in TSUBAME2.0 supercomputer failure history, there is not a single core failure~\cite{tsubame2}).
%
A \emph{failure domain} (FD) is an element of a hardware hierarchy that
can fail (e.g., a node or a cabinet). FDs constitute an FD hierarchy
(FDH) with $h$ levels. An example FDH is shown in
Figure~\ref{fig:failureHierarchy}, $h = 4$. We skip the level of single
cores because in practice the smallest FD is a node (e.g., in the TSUBAME2.0 system failure history, there are no core failures~\cite{tsubame2}).
Then, we define $\mathcal{H} = \bigcup_{1 \le j \le h} \left( \bigcup_{1
\le i \le H_{j}} H_{i,j} \right)$ to be the set of all the FD elements
in an FDH. $H_{i,j}$ and $H_j$ are element $i$ of hierarchy level $j$
and the number of such elements at level $j$, respectively. For example,
in Figure~\ref{fig:failureHierarchy} $H_{3,2}$ is the third blade (level 2) and $H_2 = 96$.

%%A \emph{failure domain} (FD) is an element of the hardware hierarchy that can fail (e.g., a node, a PSU (Power Supply Unit), a chassis, or a cabinet). FDs constitute an FD hierarchy (FDH) with $h$ levels. In an example FDH in Figure~\ref{fig:failureHierarchy}, $h = 4$. We skip the level of single cores because in practice the smallest FD is a node (e.g., in TSUBAME2.0 system failure history, there are no core failures~\cite{tsubame2}).
%%%
%%Then, we define $\mathcal{H} = \bigcup_{1 \le j \le h} \left( \bigcup_{1 \le i \le H_{j}} H_{i,j} \right)$ to be the set of all the FD elements of the hardware $\mathcal{H}$ierarchy. $H_{i,j}$ is the $i$th element of the $j$th hierarchy level. For example, in Figure~\ref{fig:failureHierarchy} $H_{3,2}$ is the third blade (level 2). $H_j$ and $|H_j|$ are the number of elements at $j$th hierarchy level and the number of single cores in every such element, respectively. For instance, if the whole machine has 1024 nodes and every node contains 32 cores, then $H_1 = 1024$ and $|H_1| = 32$ (assuming that nodes constitute level $1$).

%Finally, we also define a function $\mathscr{H}: \mathcal{H} \times \mathbb{N} \to \mathcal{H}$ that specifies relationships between respective hierarchy elements. For any $H_{i,j} \in \mathcal{H}$, $\mathscr{H}(H_{i,j}, k) = H_{l,k}$ is the element at level $k$ that contains $H_{i,j}$ (for $k > j$).

%\vspace{-0.5em}
\begin{figure}[h!]
\centering
\includegraphics[width=0.44\textwidth]{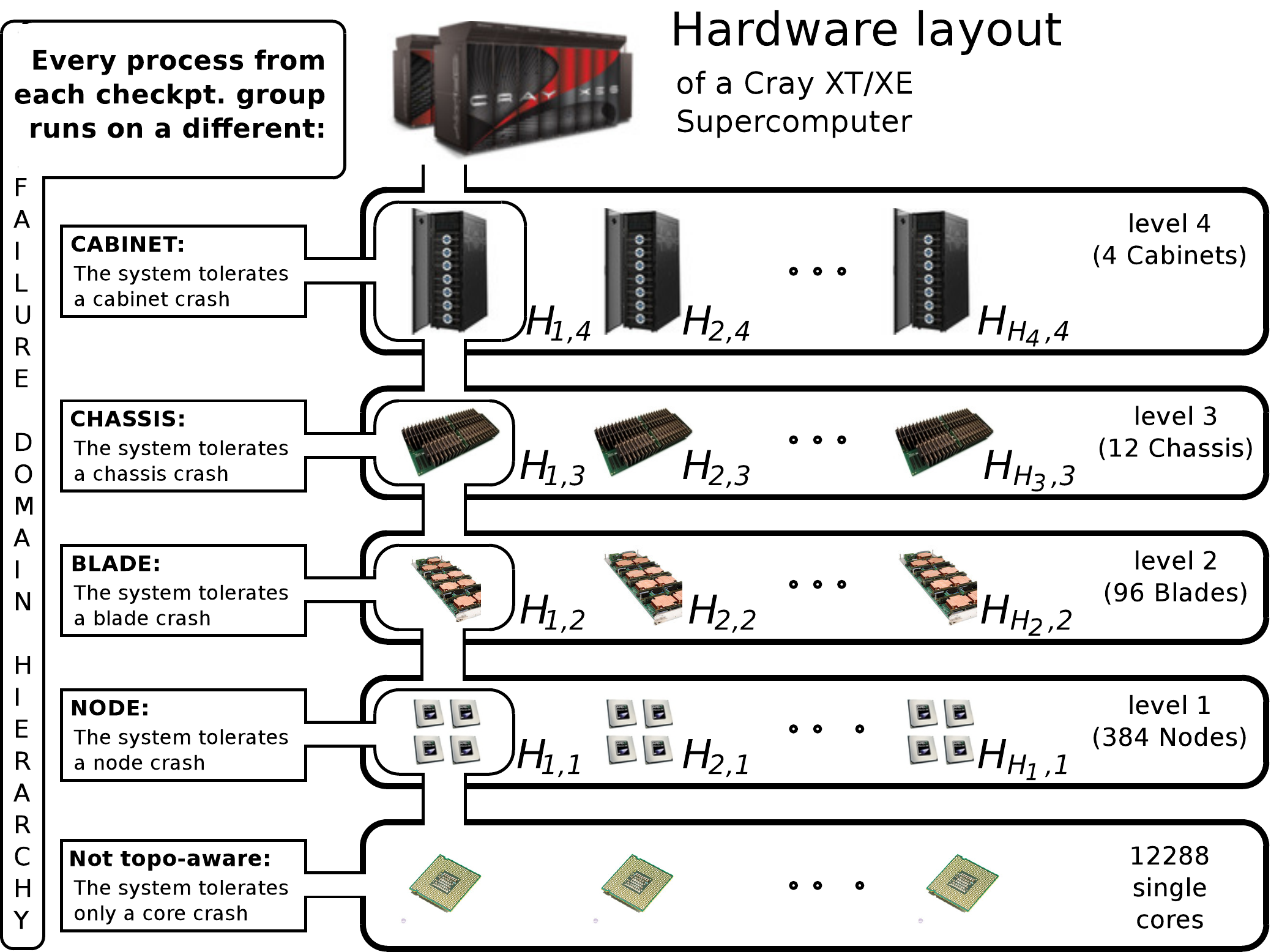}
%\vspace{-0.5em}
\caption{An example hardware layout (Cray XT/XE) and the corresponding FDH (\cref{sec:divisionIntoGroups}). In this example, $h = 4$.}
% (level 1: nodes, level 2: blades, level 3: chassis, level 4: cabinets)
\label{fig:failureHierarchy}
%\vspace{-1.0em}
\end{figure}

% \vspace{-1.0em}
%\begin{description}[leftmargin=0.6cm,style=sameline]
%\item[$l$] The number of levels of the failure domain hierarchy (for example, for the machine constructed using: nodes, racks and cabinets, $l=3$. We do not take into account the level of single cores).
% \vspace{-0.8em}
%\item[$N_i$] The number of the elements at level $i$ of the hierarchy ($1 \le i \le l$). For example, if the whole machine has 1024 compute nodes, then $N_1 = 1024$ (assuming that the compute nodes constitute the level $1$ of the hierarchy).
% \vspace{-0.8em}
%\item[$a_i$] The number of the processes that run on a single element at
%  level $i$ of the hierarchy ($1 \le i \le l$). For example, if we
%  assume that one process runs on one core, for a machine that has 32
%  cores per compute node, $a_1 = 32$
%\item[$n_i$] The number of the elements from the level $i$ of the hierarchy, used by the running application ($1 \le i \le l$). Notice that $n_i = \lceil \frac{N}{a_i} \rceil$, where $N$ is the number of processes in the application. 
%\item[$f_i$] The probability of a crash of a single element from level $i$ of the hierarchy.
%\item[$P_{i}(x_i)$] Probability of the failure of $x_i$ elements from level $i$ of the hierarchy ($1 \le x_i \le N_i$).
% \vspace{-0.8em}
%\end{description}

%\subsubsection{Groups of Processes}divisionIntoGroups
%\label{sec:groups_of_procss}

\goal{+ Describe why and how we split processes and add checksums}

\textbf{Groups of Processes }
To improve resilience, we split the process set $\mathcal{P}$ into $g$
equally-sized groups $G_i$ and
add $m$ \emph{checksum}
%, 1 \le i \le g
processes to each group to store checksums of checkpoints taken in each
group (using, e.g., the Reed-Solomon~\cite{reed1960polynomial} coding scheme).
Thus, every group can resist $m$ concurrent process crashes.
The group size is $|G| = \frac{|\mathcal{P}|}{g} + m$.

%equally-sized groups $G_i, 1 \le i \le g$.
%We add $m$ \emph{checksum}
%processes to each group to store checksums of an arbitrary coding scheme (e.g., Reed-Solomon~\cite{reed1960polynomial}).
%Thus, every group can resist $m$ concurrent process crashes.
%The group size is $|G| = \frac{|\mathcal{P}|}{g} + m$.

%
%The number of such groups and thus checkpoint processes is a model parameter.
%For notational simplicity, we divide the set of all processes
%into \emph{CheckPoint} ($CH$) processes and \emph{CoMpute} ($CM$)
%processes. The cardinalities of the sets count the number of processes:
%the number of checkpointing processes is $|CH|$ and there are $|CM| = P$
%compute processes. The size of every checkpointing group (excluding a
%single checkpointing process) is $\frac{|CM|}{|CH|}$ (we demand that
%$|CM|$ is the multiple of $|CH|$). We will discuss the grouping of $CM$
%and $CH$ processes in Section~\ref{sec:divisionIntoGroups}.

\goal{+ Show and describe the extensions to the previous system definition}

%\subsubsection{New System Definition}
\textbf{New System Definition}
We now extend the definition of a distributed system $\mathcal{D}$ to
cover the additional concepts:
%
%\begin{alignat}{2}
%\mathcal{D'} = \langle \mathcal{P}, \mathcal{E}, \mathcal{S}, \mathcal{H}, \mathcal{G}, \xrightarrow{po}, \xrightarrow{so}, \xrightarrow{hb}, \xrightarrow{co}, \mathscr{H}, \mathscr{M}\rangle
%\end{alignat}
%
%\mathcal{D'} = 
\begin{alignat}{2}
\langle \mathcal{P}, \mathcal{E}, \mathcal{S}, \mathcal{H}, \mathcal{G},
\xrightarrow{po}, \xrightarrow{so}, \xrightarrow{hb}, \xrightarrow{co},
\mathscr{M}\rangle.
\end{alignat}

\noindent
$\mathcal{G} = \{G_{1}, ..., G_{g}\}$ is a set of $\mathcal{G}$roups of processes and $\mathscr{M}: \mathcal{P} \times \mathbb{N} \to \mathcal{H}$ is a function that $\mathscr{M}$aps process $p$ to the FD at hierarchy level $k$ where $p$ runs:
$\mathscr{M}(p,k) = H_{j,k}$. $\mathscr{M}$ defines how processes are distributed over FDH. For example, if $p$ runs on blade $H_{1,2}$ from Figure~\ref{fig:failureHierarchy}, then $\mathscr{M}(p,2) = H_{1,2}$.

%Our system can experience two types of more serious crashes: \emph{catastrophic failures} and \emph{consistency failures}.

\vspace{+1.3em}
%\subsection{Catastrophic Failures}
\subsection{Handling Multiple Hardware Failures}
\label{sec:handling_multiple_hf}

\goal{+ Describe handling multiple failures with topology-awareness}

%A \emph{catastrophic failure} (we use the name from \cite{Bautista-Gomez:2011:FHP:2063384.2063427}) is a failure that can be handled only by restarting the whole computation. When utilizing XOR encoding, it happens if: (1) more than 2 processes in a checkpointing group fail, or (2) if any process fails while taking coordinated checkpoint. These failures are truly unrecoverable and we minimize the risk of scenario (1) by utilizing the topology-aware groups. Scenario (2) can be very easily alleviated by storing more than one recorded checkpoint.

More than $m$ process crashes in any group $G_i$ result in a \emph{catastrophic failure} (CF; we use the name from \cite{Bautista-Gomez:2011:FHP:2063384.2063427}) that incurs restarting the whole computation. Depending on how $\mathscr{M}$ distributes processes,
such a CF may be caused by several (or even one) crashed FDs. To minimize the risk of CFs, $\mathscr{M}$ has to be \emph{topology-aware} (t-aware): for a given level $n$ (called a \emph{t-awareness level}), no more than $m$ processes from the same group can run on the same $H_{i,k}$ at any level $k, k\le n$:
%
%To minimize the risk, we design a \emph{topology-aware} function $\mathscr{M}$ that, for a given $n$, ensures that no more than $m$ processes from the same group run on the same $H_{i,k}$ at any level $k$ such that $1 \le k \le n$.
%
%As crashes of elements at higher FDH levels are usually less probable~\cite{tsubame2}, we distribute processes at the highest possible level of the FDH. 
%
%%%Formally, to run a program in a topology-aware way at level $n$ ($1 \le n \le h$) of the FDH, $\mathscr{M}$ has to satisfy:
%
\small
\begin{alignat}{2}
&\forall p_1,p_2,...,p_m \in \mathcal{P}\quad \forall G \in \mathcal{G}\quad \forall 1 \le k \le n:\nonumber\\
&(p_1 \in G \wedge ... \wedge p_m \in G) \Rightarrow (\mathscr{M}(p_1,k) \neq ... \neq \mathscr{M}(p_m,k))
\end{alignat}
\normalsize

\noindent
Figure~\ref{fig:distribution} shows an example t-aware process distribution.

%\vspace{-0.5em}
\begin{figure}[h!]
\centering
\includegraphics[width=0.45\textwidth]{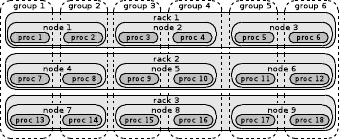}
%\vspace{-0.8em}
%\caption{An example distribution of process groups across nodes and racks with the topology-awareness at the racks level. Assume $m=1$. If rack~1 fails then all the groups can be restored to the previous consistent checkpoint. The same happens when, e.g., nodes 1, 3 and 5 go down. However, if any two racks fail simultaneously, or if nodes 1 and 4 fail, the computations have to be restarted.}
%\vspace{-0.5em}
\caption{T-aware distribution at the node \emph{and} rack level (\cref{sec:handling_multiple_hf}).}
% Assume $m=1$. If rack~1 fails then all the groups can be restored to the previous consistent checkpoint. The same happens when, e.g., nodes 1, 3 and 5 go down. However, if any two racks fail simultaneously, or if nodes 1 and 4 fail, the computations have to be restarted.}
\label{fig:distribution}
%\vspace{-0.8em}
\end{figure}

\subsection{Calculating Probability of a CF}
\label{sec:probabilityStudy}

\goal{+ Describe goals and assumptions in calculating CF probability}

%In this subsection we calculate the probability of a catastrophic failure. The reasoning is a generalization of the approach from \cite{Bautista-Gomez:2011:FHP:2063384.2063427} in a way that we construct the model basing on the whole hierarchy of hardware. The general formula for the probability that the system experiences a catastrophic failure, assuming that any $x_j$ elements from any $j$ level ($1 \le j \le l$) of the failure domain hierarchy can fail, is thus (we assume that failures at different levels are independent):

We now calculate
the probability of a catastrophic failure ($P_{cf}$) in our model. We
later (\cref{sec:reliabilityStudy}) use $P_{cf}$ to show that our
protocols are resilient on a concrete machine (the TSUMABE2.0
supercomputer~\cite{tsubame2}). If a reader is not interested in the
derivation details, she may proceed to
Section~\ref{sec:libImplementation} where we present the results. We set $m = 1$ and thus use the XOR erasure code, similar to an additional disk
in a RAID5~\cite{Chen:1994:RHR:176979.176981}. 
%Our reasoning generalizes the approach from \cite{Bautista-Gomez:2011:FHP:2063384.2063427}.
We assume that failures at different hierarchy levels are independent
and that any number $x_j$ of elements from any hierarchy level $j$ ($1
\le x_j \le H_j$, $1 \le j \le h$) can fail. Thus, 
%
%\vspace{-1.5em}
\scriptsize
\begin{alignat}{2}
P_{cf} &= \sum_{j=1}^{h} \sum_{x_j=1}^{H_j} P(x_j \cap x_{j,cf}) =
  \sum_{j=1}^{h} \sum_{x_j=1}^{H_j} P_{j}(x_j) P_{j}(x_{j,cf} | x_j).
\end{alignat}
\normalsize

%\vspace{-1.5em}
%\footnotesize
%\begin{alignat}{2}
%P_{cf} = P_{x_1,x_2,...,x_l} &= P((x_1 \cap x_{1,cf}) \cup (x_2 \cap x_{2,cf}) \cup ... \cup (x_l \cap x_{l,cf})) \nonumber\\ &= \sum_{j=1}^{l} P(x_j \cap x_{j,cf}) = \sum_{j=1}^{l} P_{j}(x_j) P_{j}(x_{j,cf} | x_j)
%\end{alignat}
%\normalsize

\goal{+ Describe the first, general formula (7)}

%\vspace{-1.5em}
%where $P_{j}(x_j)$ is the probability that $x_j$ elements of the level $j$ of the failure domain hierarchy crash
$P(x_j \cap x_{j,cf})$ is the probability that $x_j$ elements of the $j$ hierarchy level will fail \emph{and} result in a catastrophic failure. $P_{j}(x_j)$ is the probability of the failure of $x_j$ elements from level $j$ of the hierarchy. $P_{j}(x_{j,cf} | x_j)$ is the probability that $x_j$ given concurrent failures at hierarchy level $j$ are catastrophic to the system. It is difficult to analytically derive $P_{j}(x_j)$ as it is specific for every machine. For our example study (see Section~\ref{sec:reliabilityStudy}) we use the failure rates from the TSUBAME2 failure history \cite{tsubame2}.

\goal{+ Describe the conditional formula (8) and show the final Equation (9)}

In contrast, $P_{j}(x_{j,cf} | x_j)$ can be calculated using combinatorial theory. Assume that $\mathscr{M}$ distributes processes in a t-aware way at levels $1$ to $n$ of the FDH ($1 \le n \le h$). First, we derive $P_{j}(x_{j,cf} | x_j)$ for any level $j$ such that $1 \le j \le n$:

%\cite{Sato:2012:DMN:2388996.2389022, Bautista-Gomez:2011:FHP:2063384.2063427}

%%\vspace{-2.0em}
%\scriptsize
%\begin{alignat}{2}
%P_{j}(x_{j,cf} | x_j) &= \frac{\binom{|G|}{2}\Big\lceil\frac{|H_j|}{|G|}\Big\rceil\binom{|H_j|-2}{x_j-2} }{ \binom{|H_j|}{x_j} }
%\end{alignat}
%\normalsize

%\vspace{-2.0em}
\scriptsize
\begin{alignat}{2}
P_{j}(x_{j,cf} | x_j) &= \frac{D_j \cdot \binom{|G|}{2}\cdot
\binom{H_j-2}{x_j-2} }{ \binom{H_j}{x_j} }.
\end{alignat}
\normalsize

%\vspace{-1.5em}
\noindent
$\binom{|G|}{2}$ is the number of the possible catastrophic failure
scenarios \emph{in a single group} ($m=1$ thus any two process crashes in
one group are catastrophic). $D_j$ is the number of such single-group
scenarios \emph{at the whole level $j$} and is equal to $\left\lceil
\frac{H_j}{|G|} \right\rceil$ (see Figure~\ref{fig:dist_expl} for intuitive
explanation). $\binom{H_j-2}{x_j-2}$ is the number of the remaining
possible failure scenarios and $\binom{H_j}{x_j}$ is the total number of
the possible failure scenarios. Second, for remaining levels $j$ ($n+1
\le j \le h$) $\mathscr{M}$ is \emph{not} t-aware and thus in the
worst-case scenario any element crash is catastrophic: $P_{j}(x_{j,cf} |
x_j) = 1$. The final formula for $P_{cf}$ is thus
%we assume the worst-case scenario in which the processes are distributed over the whole machine and the failure of any element at any level $> n$ is catastrophic. The final formula for $P_{cf}$ is thus:

%To obtain the number of such scenarios in the whole machine, we cannot simply multiply $\binom{w}{2}$ by the number of checkpointing groups ($|G|$), because the result may be bigger than the number of the elements of the $j$ hardware hierarchy level.

%\vspace{-1.8em}
\scriptsize
\begin{alignat}{2}
P_{cf} &= \sum_{j=1}^{n} \sum_{x_j=1}^{H_j} P_{j}(x_j) \frac{D \cdot
\binom{|G|}{2} \cdot \binom{H_j-2}{x_j-2} }{ \binom{H_j}{x_j} }
+\sum_{j=n+1}^{h} \sum_{x_j=1}^{N_j} P_{j}(x_j).
\end{alignat}
\normalsize

%\vspace{-0.5em}
\begin{figure}[h!]
\centering
\includegraphics[scale=1.7]{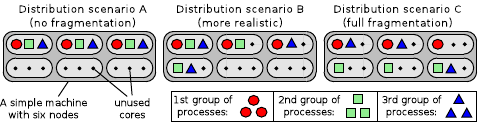}
%\vspace{-0.5em}
%%\caption{
%%Consider a very simple machine (big rounded rectangle) that consists only of nodes (smaller rounded rectangles). We have $h=1$ and $H_1 = 6$. Circles, squares and triangles symbolize the processes belonging to three different process groups ($|G|=3$ and $g=3$). Black dots are unused cores. Scenario A is an example best-case distribution pattern (full utilization of the nodes), Scenario B is an example random pattern (partial utilization), Scenario C is the worst-case pattern (full fragmentation)}
%%%As we want to get the upper bound for the $P_{cf}$, we use Scenario C and thus we arrive at factor $\lceil N_j / w \rceil$ in formula (2).}
\caption{(\cref{sec:probabilityStudy}) Consider three process distribution scenarios by $\mathscr{M}$ (\emph{each} is t-aware). Optimistically, processes can be distributed contiguously (scenario A) or partially fragmented (scenario B). To get the upper bound for $P_{cf}$ we use the worst-case pattern (scenario C). Now, to get the number of single-group CF scenarios at the whole level $j$ ($D_j$), we need to obtain the number of the groups of \emph{hardware elements} at $j$ that
hold process groups: $\lceil H_j/|G| \rceil$.}
%
%As we want to get the upper bound for the $P_{cf}$, we use Scenario C and thus we arrive at factor $\lceil N_j / w \rceil$ in formula (2).}
\label{fig:dist_expl}
%\vspace{-0.5em}
\end{figure}

\section{Holistic Resilience Protocol}
\label{sec:libImplementation}

\goal{Describe the goal and the implementation details of the protocol}

We now describe an example conceptual implementation of holistic fault tolerance for RMA
that we developed to understand the tradeoffs between the resilience and performance in RMA-based systems. 
We implement it as a portable library (based on C
and MPI) called \flib{}. We utilize MPI-3's one sided
interface, but any other RMA model enabling relaxed memory
consistency could be used instead (e.g., UPC or Fortran 2008). We use
the publicly available \fompi{} implementation of MPI-3 one sided as MPI
library~\cite{fompi} but any other MPI-3 compliant
library would be suitable. 
For simplicity we assume that the user application uses one
contiguous region of shared memory of the same size at each process.
Still, all the conclusions drawn are valid for any other
application pattern based on RMA. Following the MPI-3 specification, we call this shared region of memory at every process a $window$.
%Still, we provide the interface that
%enables constructing the hierarchy of failure domains with an arbitrary
%number of levels. The interface can be easily implemented with the
%machine-specific code.
%
Finally, we divide user processes (referred to as CoMputing processes, $CMs$) into groups (as described in Section~\ref{sec:divisionIntoGroups})
and add one CHecksum process (denoted as $CH$) per group ($m=1$). For any computing process $p$, we denote the $CH$ in its group as $CH(p)$. $CHs$ store and update XOR checksums of their $CMs$.

\subsection{Protocol Overview}
\label{sec:prot_over_small}

\goal{+ Describe protocol layers and modules}

In this section we provide a general overview of the layered protocol implementation
(see Figure~\ref{fig:generalOverview}). The first part (layer 1) logs accesses.
The second layer takes uncoordinated checkpoints (called \emph{demand} checkpoints) to trim the logs.
Layer 3 performs regular coordinated checkpoints. All layers are diskless.
Causal recovery replays memory accesses. Finally, our FDH increases resilience of the whole protocol.
%
%Layer 1 and 2 together constitute the uncoordinated part of the protocol. In case the logged application has intensive communication patterns which introduce too much logging overhead, the whole protocol stops recording $puts$ and $gets$ and falls back to coordinated checkpointing. The worst-case scenario of our protocol is thus the well-known and commonly-used pattern of regular coordinated checkpoints.
%
%Finally, in an orthogonal approach, the whole protocol uses the hierarchy of failure domains to improve resiliency.

 %\vspace{-0.7em}
\begin{figure}[h!]
\centering
\includegraphics[width=0.4\textwidth]{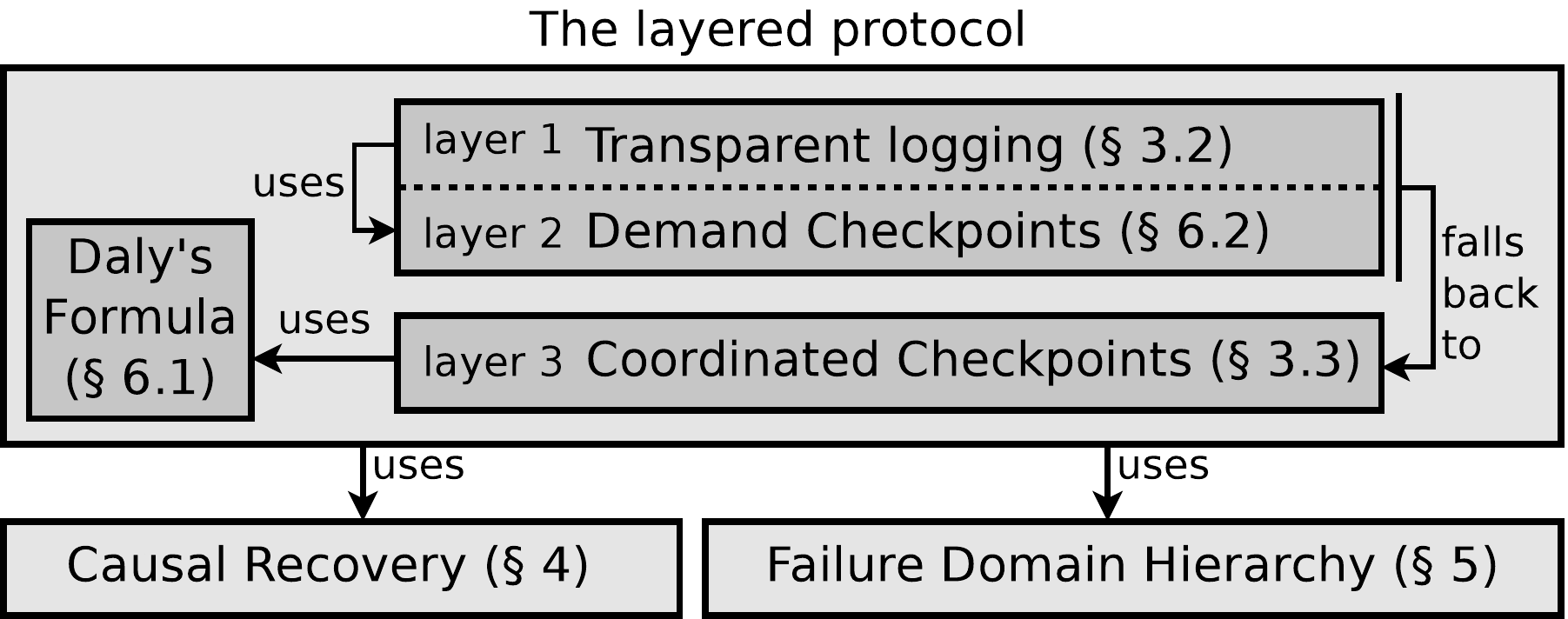}
 %\vspace{-0.8em}
\caption{The overview of the protocol (\cref{sec:prot_over_small}). Layer 1 and 2 constitute the
 uncoordinated part of the protocol that falls back to the
 coordinated checkpointing if logging fails or if its overhead is too high.}
\label{fig:generalOverview}
%\vspace{-0.5em}
\end{figure}

\begin{figure*}
\centering
\vspace{-1.2em}
% \subfloat[Collisions constitute $\approx$1\% of all the inserts.]{
 \subfloat[Distribution of node crashes (samples and the fit) (\cref{sec:reliabilityStudy}).]{
  \includegraphics[width=0.23\textwidth]{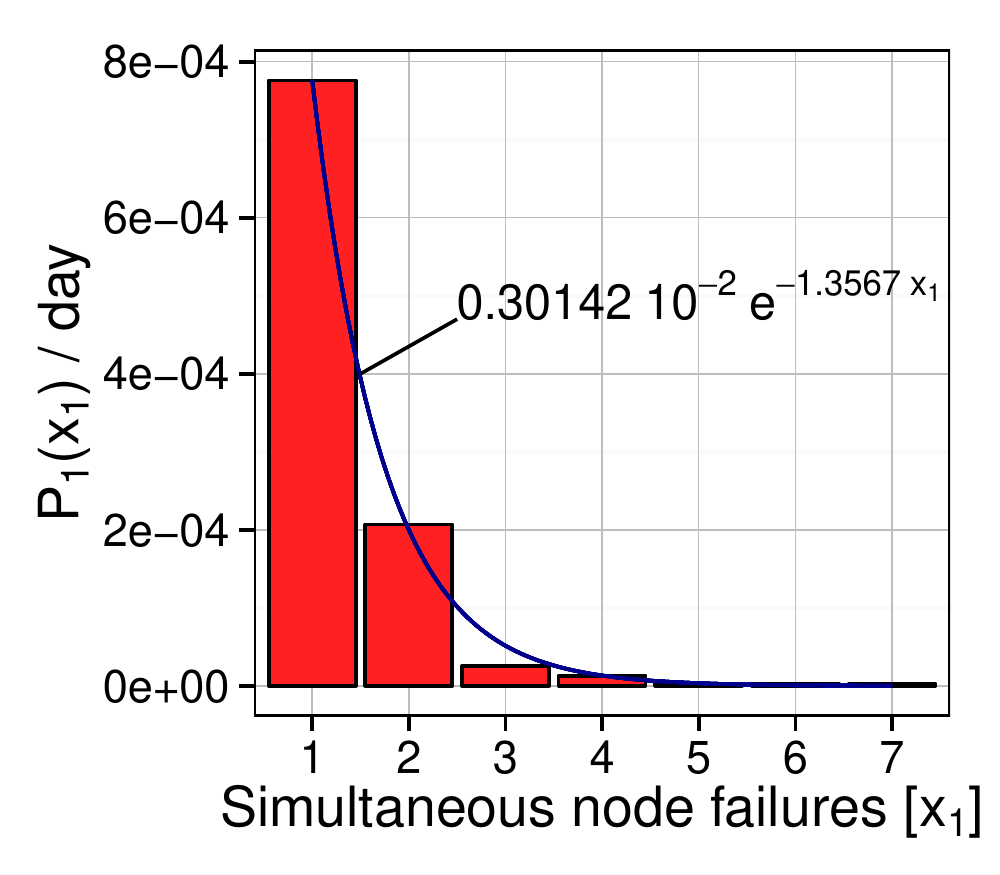}
  \label{fig:pdfRacks}
 }\hfill
 \subfloat[Distribution of PSU crashes (samples and the fit) (\cref{sec:reliabilityStudy}).]{
  \includegraphics[width=0.23\textwidth]{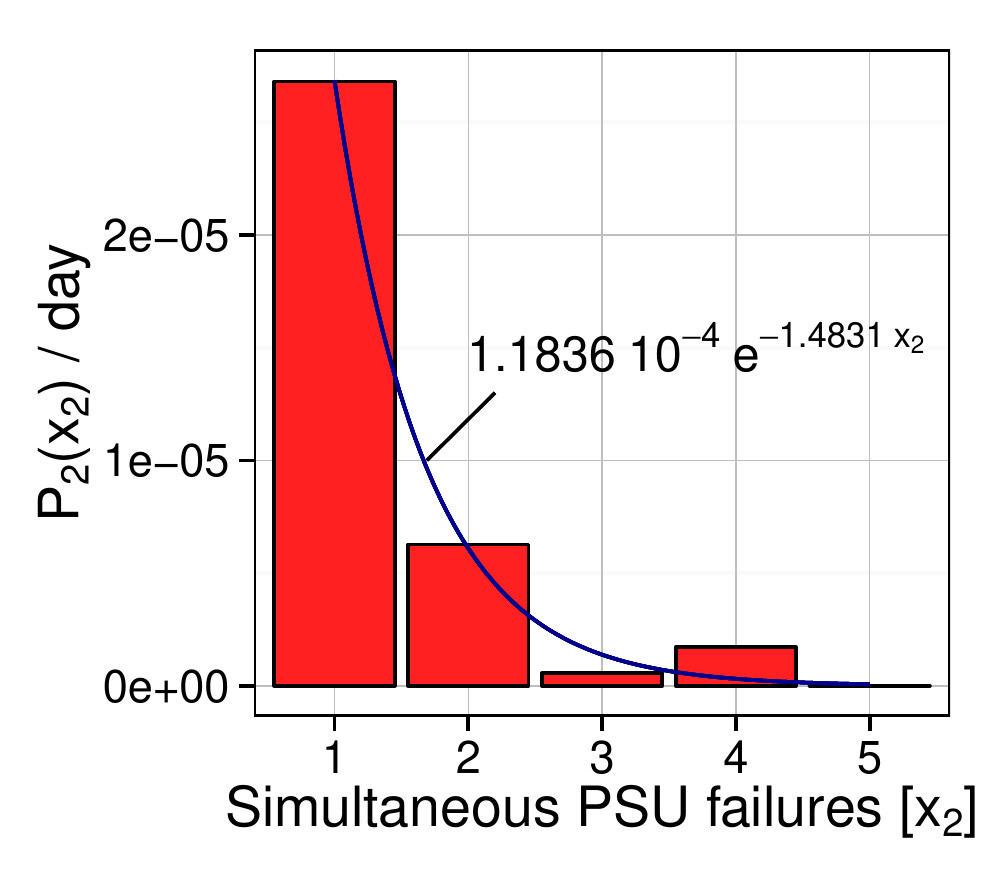}
  \label{fig:pdfNodes}
 }\hfill
 \subfloat[Probability of a catastrophic failure (\cref{sec:compar_resil}).]{
  \includegraphics[width=0.23\textwidth]{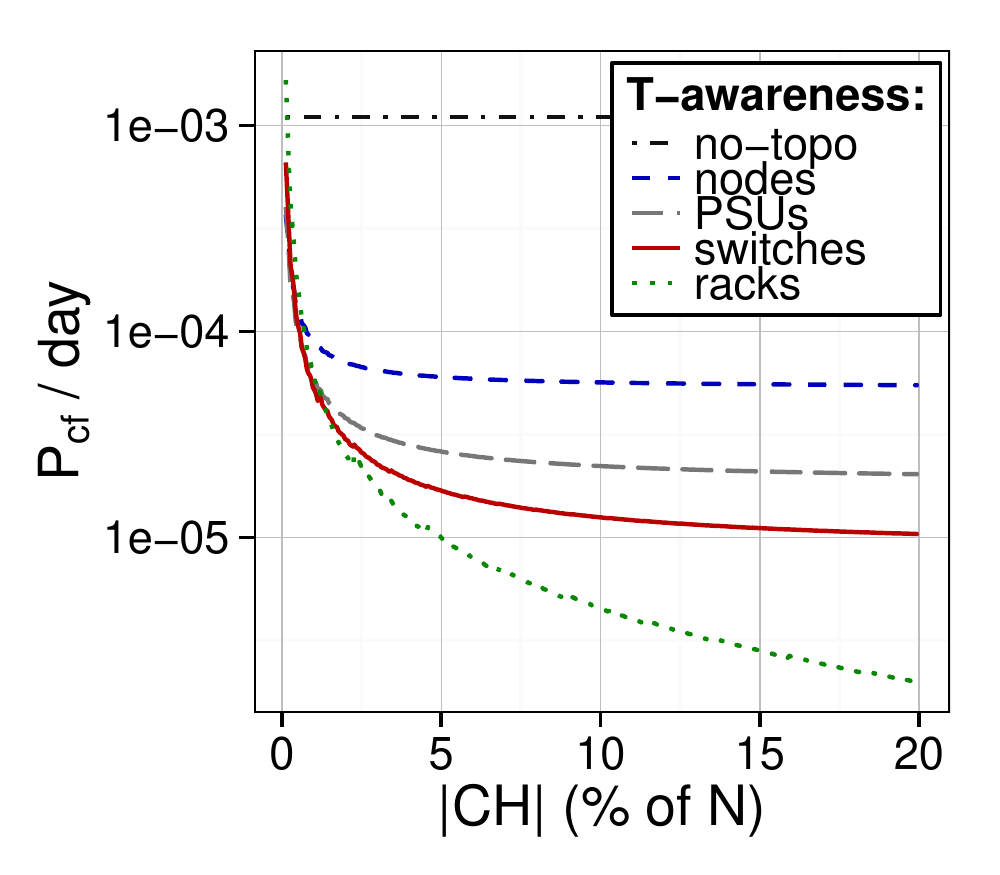}
  \label{fig:probabilityPict}
 }\hfill
  \subfloat[NAS FFT (class C) fault-free runs: checkpointing (\cref{sec:nas_eval}).]{
  \includegraphics[width=0.23\textwidth]{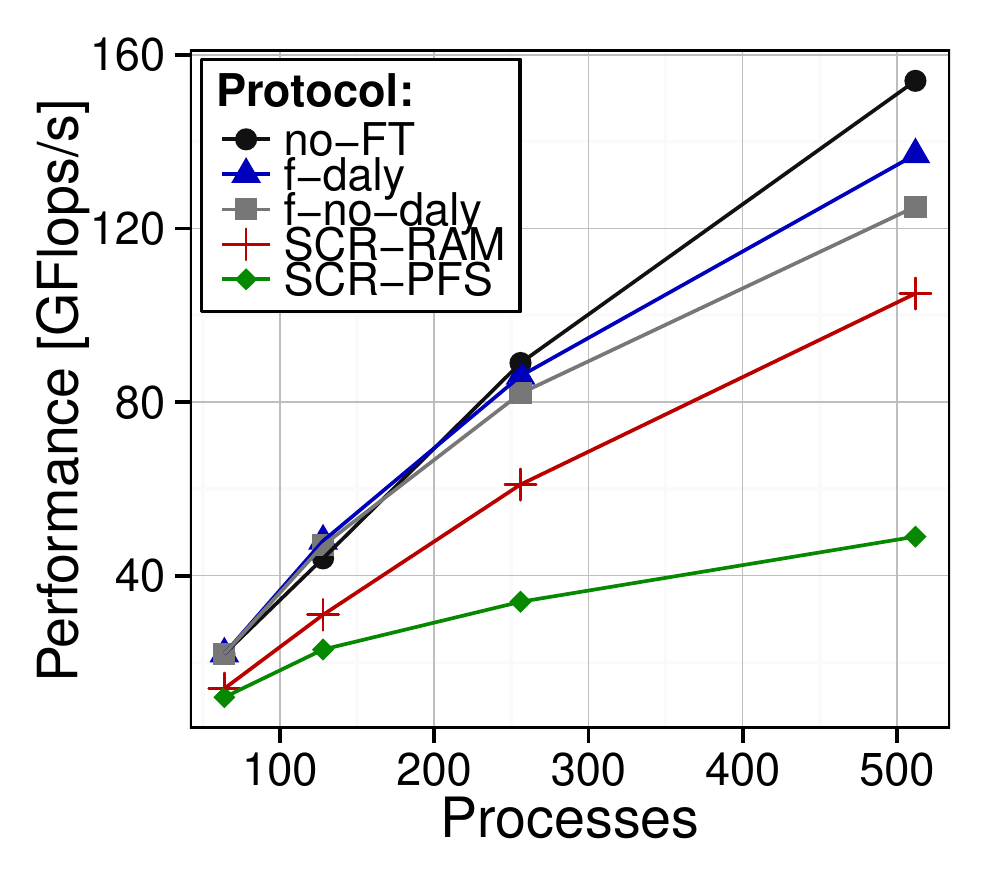}
  \label{fig:fftFFrun}
 }\hfill
% \subfloat[Access Counter results (\cref{sec:eval_ac})]{
%  \includegraphics[width=0.23\textwidth]{nodes-barplot-e-eps-converted-to.pdf}
%  \label{fig:counter_perf}
% }
% \hfill
% \subfloat[Varying $R_{cols}$ for $P=256$.]{
%  \includegraphics[width=0.18\textwidth]{procs-ins-new_P-256-eps-converted-to.pdf}
%  \label{fig:ht_inserts_P-256}
% }
%\vspace{-0.5em}
\caption{Distribution of PSU \& node failures, $P_{cf}$ in TSUBAME2.0
running 4,000 processes, and the performance of NAS 3D FFT.}
\label{fig:probability}
\vspace{-1.0em}
\end{figure*}

\goal{+ Say how and why we use Daly's formula}

\textbf{Daly's Interval }
Layer 3 uses Daly's formula~\cite{Daly:2006:HOE:1134241.1134248} as the optimum interval between coordinated checkpoints: \small$\sqrt{2 \delta M} \cdot [1 + 1/3 \sqrt{\delta/(2 M)} + (1/9) (\delta/(2 M))] - \delta$ \normalsize(for $\delta < 2 M$), or $M$ (for $\delta \ge 2 M$). $M$ is the MTBF (mean time between failures that \flib{} handles with coordinated checkpointing) for the target machine and $\delta$ is the time to take a checkpoint. The user provides $M$ while $\delta$ is estimated by our protocol.

%\goal{+ Say how and why we degenerate to coordinated checkpoints}
%
%\textbf{Protocol Switch }
%%
%\flib{} gathers the
%information on free memory amount left at
%processes and the number of issued remote accesses. 
%It then estimates the logging memory/performance overhead and, in case the
%overhead is too high, it stops logging and only takes coordinated
%checkpoints. \htor{sounds like magic}
%%%We select one designated process $d$ to gather the
%%%information on free memory left at
%%%processes and the average/minimum number of issued remote accesses (per process). 
%%%Using this data $d$ decides if logging incurs too much overhead. If yes,
%%%the protocol degenerates to the coordinated checkpointing scheme.
%Basing on the knowledge of the specific applications, the thresholds can
%be accordingly adjusted by the user.

%\vspace{-1.0em}
%\subsection{Interfacing with User and Runtime}
%\label{sec:implDetails}

\goal{+ Say how we interface with user/runtime}

\textbf{Interfacing with User Programs and Runtime }
%
%We now describe how \flib{} integrates with user programs.
%
\flib{} routines are called after each RMA action. 
This would entail runtime system calls in compiled
languages and we use the PMPI profiling interface~\cite{mpi3} in our
implementation.  
During window creation
the user can specify: (1) the number of $CHs$,
(2) MTBF, (3) whether to use topology-awareness. 
After window creation, the protocol divides processes into $CMs$ and
$CHs$. If the user enables t-awareness, groups of processes running on the same
FDs are also created. In the current version \flib{} takes into account computing nodes
when applying t-awareness.

\subsection{Demand Checkpointing}

\goal{+ Motivate and describe coordinated checkpoints}

\emph{Demand checkpoints} address the problem of diminishing amounts of
memory per core in today's and future computing centers.
%
%Since logs can be discarded once they are included in an uncoordinated
%checkpoint, a {demand checkpoint} can be requested by processes to
%empty their local logs. For example, if a process $p$ detects that its
%local log is full, then it can request another process $q$ to take a
%checkpoint and discard logs for process $q$ once the checkpoint is
%done. 
%
If free memory at $CM$ process $p$ is scarce, $p$ selects the process $q$ with
the largest $LP_p[q]$ or $LG_p[q]$ and requests a demand checkpoint.\htor{? cUsing
this data iheckpoint}
First, $p$ sends a \emph{checkpoint request} to $CH(q)$ which, in turn, forces $q$ to
checkpoint. This can be done by: closing all the epochs, locking all the relevant data
structures, calculating the XOR checksum, and: (1) streaming the result
to $CH(q)$ piece by piece or (2) sending the
result in one bulk. $CH(q)$ integrates the received checkpoint data into the existing XOR checksum. Variant (1) is memory-efficient, and (2) is less
time-consuming. Next, $q$ unlocks all the data
structures. Finally, $CH(q)$ sends a confirmation with the epoch number $E$\ptoq\ and respective counters ($GNC_q$, $GC_q$, $SC_q$) to $p$.
Process $p$ can delete logs of actions $a$ where $a.EC < E$\ptoq\@, $a.GNC < GNC_q$, $a.GC < GC_q$, $a.SC < SC_q$.
%

%%\vspace{-1.0em}
%\subsection{Protocol Switch}
%\label{sec:implDetails}
%
%%gathers the statistics about remote operations issued by $CMs$. Before every coordinated checkpoint, it collects 
%In this section we discuss how our protocol switches to the coordinated checkpointing. In our implementation, one designated process gathers three types of
%information (per checkpointing interval) that describe the
%communication. These are: the amounts of free logging space left at
%processes, the average number of issued remote operations (per process)
%and the minimum number of issued operations. It then compares these
%values with the previous ones and, basing on several fixed thresholds,
%decides whether to degenerate to the coordinated checkpoints scheme.
%Basing on the knowledge of the specific applications, the thresholds can
%be accordingly adjusted by the user.

%\vspace{-0.7em}
\section{Testing and Evaluation}
\label{sec:testing}

\goal{Summarize the section, explain the notation used}

In this section we first analyze the resilience of our protocol using real data from TSUBAME2.0~\cite{tsubame2} failure history. Then, we test the performance of \flib{} with a NAS benchmark~\cite{Bailey91thenas} that computes 3D Fast Fourier Transformation and a distributed key-value store. 
We denote the number of $CHs$ and $CMs$ as $|CH|$ and $|CM|$, respectively.

%\vspace{-1.0em}
\subsection{Analysis of Protocol Resilience}
\label{sec:reliabilityStudy}

\goal{+ Explain why we do resilience analysis}

%In this section we study the resilience of our protocol. The purpose of
%this analysis is to calculate the probability of a catastrophic failure
%$P_{cf}$ and to show how this value changes when applying the
%topology-aware partitioning at different levels of hardware hierarchy.
%For this purpose, we will use Equation (9) from
%Section~\ref{sec:probabilityStudy}.

%%\vspace{-1.1em}
%\subsubsection{TSUBAME2.0}

%We base our study on the TSUBAME2.0~\cite{tsubame2} supercomputer. We set $H_j$
%according to the machine parameters~\cite{Sato:2012:DMN:2388996.2389022}
%and we model the failure domain hierarchy with 4 levels: node level,
%PSU (Power Supply Unit) level, edge switch level and rack level
%\cite{Sato:2012:DMN:2388996.2389022}. In order to calculate $P_{cf}$ we
%first need to obtain the distribution $P_{j}(x_j)$ (it specifies the probability of $x_j$ simultaneous failures of the elements from level $j$ of the hardware hierarchy). In addition, we need to calculate $P_{j}(x_{j,cf} | x_j)$ (it determines the probability that a crash of $x_j$ elements from hardware hierarchy level $j$ is catastrophic). We
%calculate $P_{j}(x_{j,cf} | x_j)$ using combinatorial theory
%(Equation~(2)). To obtain $P_j(x_j)$ we studied the history of
%TSUBAME2.0 failures, available at~\cite{tsubame2}. We analyzed 1962
%faults that occurred between 2010/10/31 and 2013/04/11. Following
%\cite{Bautista-Gomez:2011:FHP:2063384.2063427} we decided to use
%exponential probability distributions, where the argument is the number
%of simultaneous failures. 

Our protocol stores all data in volatile memories to avoid I/O performance penalties
and frequent disk and parallel file system failures~\cite{Sato:2012:DMN:2388996.2389022,disk_fails}.
This brings several questions on whether the scheme is resilient in practical environments.
To answer this question, we calculate the probability of a catastrophic failure
$P_{cf}$ (using Equations~(7) and~(9)) of our protocol, applying t-awareness at different levels of FDH.

\goal{+ Describe how we use our model with data from TSUBAME}

We first fix model parameters ($H_j$, $h$) to reflect the hierarchy of TSUBAME2.0.
TSUBAME2.0 FDH has 4 levels~\cite{Sato:2012:DMN:2388996.2389022}: nodes,
power supply units (PSUs), edge switches, and racks ($h=4$)~\cite{Sato:2012:DMN:2388996.2389022}.
Then, to get $P_{cf}$, we calculate distributions $P_j(x_j)$ that determine the probability
of $x_j$ concurrent crashes at level $j$ of the TSUBAME FDH.
To obtain $P_j(x_j)$ we analyzed 1962
crashes in the history of
TSUBAME2.0 failures~\cite{tsubame2}. Following
\cite{Bautista-Gomez:2011:FHP:2063384.2063427} we decided to use
exponential probability distributions, where the argument is the number
of concurrent failures $x_j$. 
We derived four probability density functions (PDFs) that approximate the failure distributions of
nodes (\small$0.30142 \cdot 10^{-2} e^{-1.3567 x_1}$\normalsize),
PSUs (\small$1.1836 \cdot 10^{-4} e^{-1.4831 x_2}$\normalsize),
switches (\small$3.9249 \cdot 10^{-5} e^{-1.5902 x_3}$\normalsize),
and racks (\small$3.2257 \cdot 10^{-5} e^{-1.5488 x_4}$\normalsize).
The unit is failures per day. 
Figures~\ref{fig:pdfRacks} and~\ref{fig:pdfNodes}
illustrate two PDF plots with histograms.
The
distributions for PSUs, switches, and racks are based on real data
only. For nodes it was not always possible to determine the exact
correlation of failures. Thus, we pessimistically assumed
(basing on~\cite{Bautista-Gomez:2011:FHP:2063384.2063427}) that single
crashes constitute 75\% of all node failures, two
concurrent crashes constitute 20\%, and other values decrease
exponentially.

\subsubsection{Comparison of Resilience}
\label{sec:compar_resil}

Figure~\ref{fig:probabilityPict} shows the resilience of our protocol when using
five t-awareness strategies. The number of processes $N$ is 4,000.
$P_{cf}$ is normalized to one day period. Without t-awareness (\smalltt{no-topo})
a single crash of any FD of TSUBAME2.0 is catastrophic, thus $P_{cf}$ does
not depend on $|CH|$.
In other scenarios every process from
every group runs on a different node (\smalltt{nodes}), PSU ({\smalltt{PSUs}}),
switch enclosure ({\smalltt{switches}}) and rack (\smalltt{racks}).
In all cases $P_{cf}$ decreases proportionally to the increasing $|CH|$, however at some point the
exponential distributions ($P_j(x_j)$) begin to dominate the results.
Topology-awareness at higher hierarchy levels
significantly improves the resilience of our protocol.
For example, if $CH = 5\% N$, $P_{cf}$ in the \smalltt{switches}
scenario is $\approx$4 times lower than in \smalltt{nodes}.
Furthermore, all t-aware schemes are 1-3 orders of
magnitude more resilient than \smalltt{no-topo}.

\goal{+ State our schemes are safe and we don't want I/O and disks}

The results show that even a simple scheme (\smalltt{nodes}) significantly
improves the resilience of our protocol that performs
only in-memory checkpointing and logging. We conclude that
costly I/O flushes to the parallel file system (PFS) are not required for
obtaining a high level of resilience. 
On the contrary, such flushes may
even \emph{increase} the risk of failures. They usually entail stressing the I/O system for significant amounts of time
\cite{Sato:2012:DMN:2388996.2389022}, and stable storage is often the
element most susceptible to crashes. For example, a Blue Gene/P
supercomputer had 4,164 disk fail events in 2011 (for 10,400 total
disks)~\cite{disk_fails}, and its PFS failed 77 times, almost two times
more often than other hardware~\cite{disk_fails}.

\begin{figure*}
\centering
%\vspace{-1.2em}
% \subfloat[NAS FFT (class C) recovery from s demand checkpoint.]{
%  \includegraphics[width=0.23\textwidth]{nas-rec-eps-converted-to.pdf}
%  \label{fig:nas-rec}
% }\hfill
  \subfloat[NAS FFT (class A) fault-free runs: demand checkpointing.]{
  \includegraphics[width=0.30\textwidth]{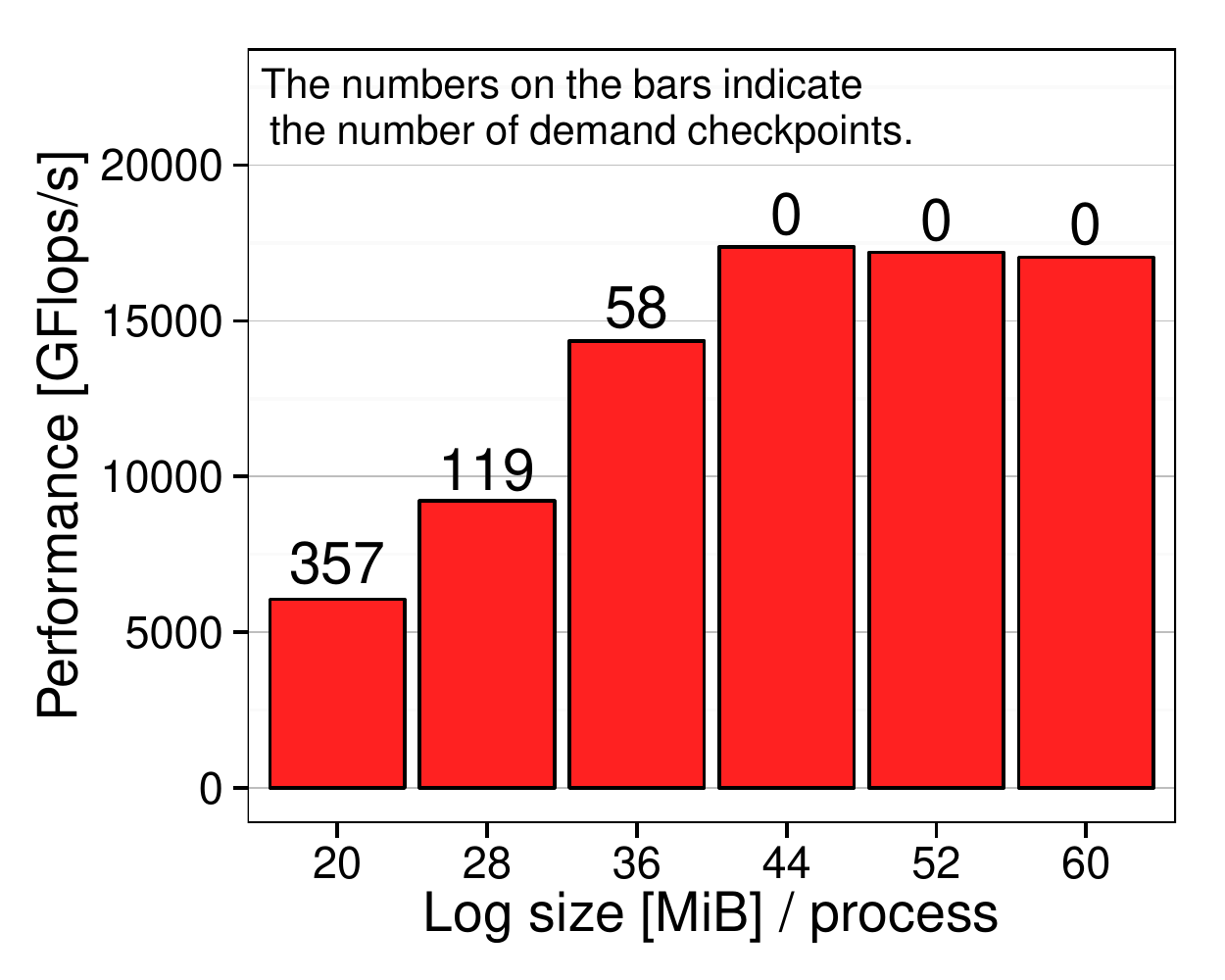}
  \label{fig:dchecks}
 }\hfill
% \subfloat[Collisions constitute $\approx$1\% of all the inserts.]{
 \subfloat[NAS FFT (class A) fault-free runs: logging.]{
  \includegraphics[width=0.3\textwidth]{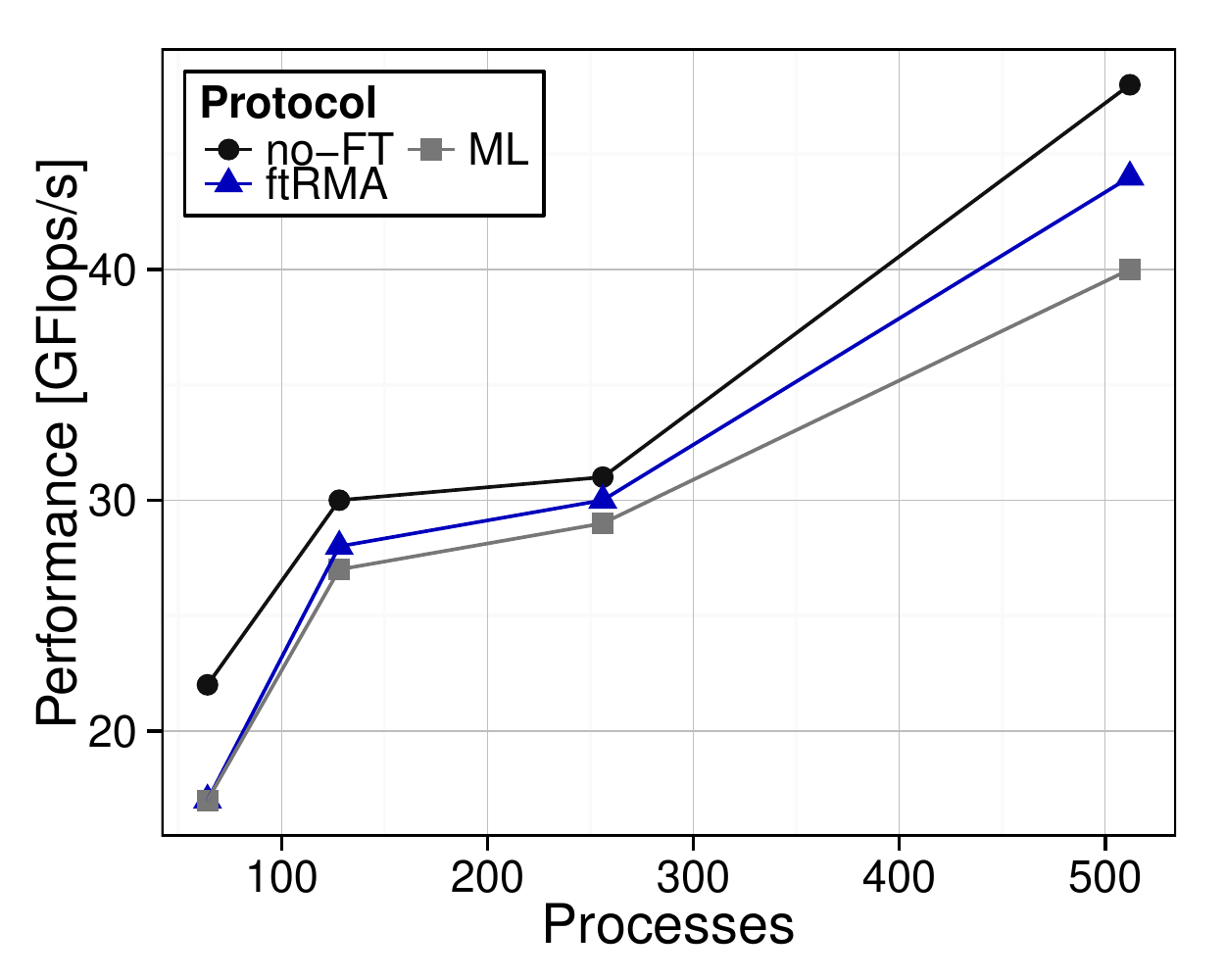}
  \label{fig:nas-logs}
 }\hfill
 \subfloat[Key-value store fault-free runs.]{
  \includegraphics[width=0.3\textwidth]{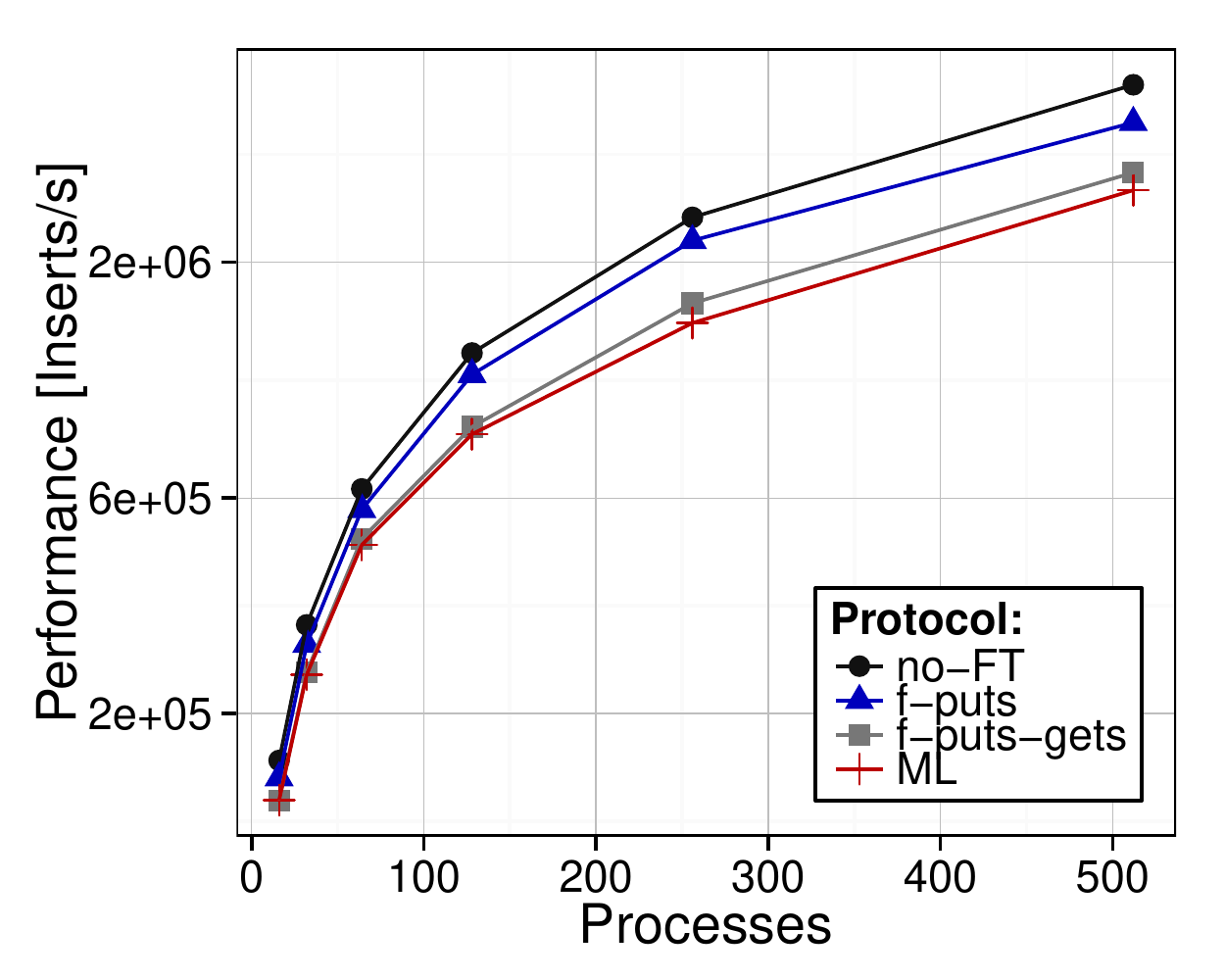}
  \label{fig:hashFFrun}
 }
 %\hfill
%  \subfloat[NAS FFT (class A) fault-free runs: logging (\cref{sec:nas_eval}).]{
%  \includegraphics[width=0.23\textwidth]{nas-logs-eps-converted-to.pdf}
%  \label{fig:nas-logs}
% }\hfill
% \subfloat[Access Counter results (\cref{sec:eval_ac})]{
%  \includegraphics[width=0.23\textwidth]{nodes-barplot-e-eps-converted-to.pdf}
%  \label{fig:counter_perf}
% }
% \hfill
% \subfloat[Varying $R_{cols}$ for $P=256$.]{
%  \includegraphics[width=0.18\textwidth]{procs-ins-new_P-256-eps-converted-to.pdf}
%  \label{fig:ht_inserts_P-256}
% }
%\vspace{-0.5em}
\caption{Performance of the NAS FFT code (\cref{sec:nas_eval}) and the key-value store (\cref{sec:kv_eval}).}
\label{fig:res-fft-kv}
%\vspace{-1.0em}
\end{figure*}

%\vspace{-1.0em}
\subsection{Analysis of Protocol Performance}

%We now discuss the performance of our adaptive protocol after the
%integration with two applications: NAS 3D FFT and a distributed
%key-value store. We analyze the performance of fault-free
%runs and the recovery process. Integrating our \flib{} implementation
%with the application code was trivial and required minimal code changes
%resulting in the same code complexity. We execute all benchmarks on a
%large scale Cray Supercomputer (Blue Waters). It consists of several
%types of computing nodes, however for our experiments we use only XE6
%nodes. Each compute node contains four 8-core 2.3~GHz AMD Opteron 6276
%(Interlagos) and is connected to other nodes through a 3D-Torus
%Gemini network. We use the Cray Programming Environment 4.1.40 to
%compile the code.

\goal{+ Describe the evaluation section, state we're doing worst-case benchmarks}

%of fault-free runs and the recovery scheme of

We now discuss the performance of our fault tolerance protocol after the
integration with two applications: NAS 3D FFT and a distributed
key-value store. Both of these
applications are characterized by intensive communication patterns,
thus they demonstrate worst-case scenarios for our protocol. Integrating \flib{}
with the application code was trivial and required minimal code changes
resulting in the same code complexity.

\goal{+ Describe SCR and how we configure it}

%\vspace{-2.0em}
\textbf{Comparison to Scalable Checkpoint/Restart }
We compare \flib{} to Scalable Checkpoint-Restart (SCR)~\cite{scr}, a popular open-source message passing library that provides
checkpoint and restart capability for MPI codes but does not enable logging.
We turn on the XOR scheme in SCR
and we fix the size of SCR groups~\cite{scr} so that they match the analogous parameter
in \flib{} ($|G|$).
To make the comparison fair,
we configure SCR to save checkpoints to both in-memory tmpfs (\smalltt{SCR-RAM})
and to the PFS
(\smalltt{SCR-PFS}). 
%Blue Waters does not provide local node disks.

\goal{+ Describe a simple ML protocol that we compare to }

%\vspace{-2.0em}
\textbf{Comparison to Message Logging }
To compare the logging overheads in MP and RMA we also developed a simple message logging (ML) scheme (basing on the protocol from~\cite{Riesen:2012:ASI:2388996.2389021}) that
records accesses.
Similarly to~\cite{Riesen:2012:ASI:2388996.2389021}
we use additional processes to store protocol-specific access logs;
the data is stored at the sender's or receiver's side
depending on the type of operation.
%%%To compare the logging overheads in MP and RMA we also developed a simple scheme 
%%%that records accesses using an MP protocol based on $ml$MPI, a recent library described
%%%in~\cite{Riesen:2012:ASI:2388996.2389021} that logs messages. 
%%%Similarly to~\cite{Riesen:2012:ASI:2388996.2389021}
%%%we use additional processes to store protocol-specific logs;
%%%the data is stored at the sender's or receiver's side
%%%depending on the type of operation.
%we record the data at the sender's side\\

%Similarly to~\cite{Riesen:2012:ASI:2388996.2389021} we use additional
%processes to store protocol-specific logs, and the data is recorded at the sender's side.

\goal{+ Describe hardware we used for benchmarks}

We execute all benchmarks on the
Monte Rosa system and we use Cray XE6 computing
nodes. Each node contains four 8-core 2.3~GHz AMD Opterons 6276
(Interlagos) and is connected to a 3D-Torus
Gemini network. We use the Cray Programming Environment 4.1.46 to
compile the code.

%\vspace{-1.0em}
\subsubsection{NAS 3D Fast Fourier Transformation}
\label{sec:nas_eval}

\goal{+ Describe NAS 3D FFT}

Our version of the NAS 3D FFT~\cite{Bailey91thenas} benchmark is based on MPI-3
nonblocking \textsc{put}s (we exploit the overlap of computation and
communication). The benchmark calculates 3D FFT using a 2D decomposition.

%(class C code~\cite{Bailey91thenas})

\goal{+ Describe the performance of fault-free runs when checkpointing}

\textbf{Performance of Coordinated Checkpointing }
We begin with evaluating our checkpointing ``Gsync'' scheme.
Figure~\ref{fig:fftFFrun} illustrates the performance of NAS FFT fault-free
runs. We compare: the original application code without any fault-tolerance
(\smalltt{no-FT}), \flib{}, \smalltt{SCR-RAM}, and \smalltt{SCR-PFS}. We fix $|CH| = 12.5\% |CM|$.
We include two \flib{} scenarios: \smalltt{f-daly} (use Daly's formula for coordinated checkpoints), and
\smalltt{f-no-daly} (fixed frequency of checkpoints without Daly's formula, $\approx$2.7s for 1024 processes).
We use the same t-awareness policy in all codes (\smalltt{nodes}).
The tested schemes have the respective fault-tolerance
overheads over the baseline \smalltt{no-FT}: 1-5\% (\smalltt{f-daly}), 1-15\% (\smalltt{f-no-daly}), 21-37\% (\smalltt{SCR-RAM}) and
46-67\% (\smalltt{SCR-PFS}). The performance
of SCR-RAM is lower than \smalltt{f-daly} and \smalltt{f-no-daly} because \flib{} is based on the Gsync scheme that incurs less synchronization. \smalltt{SCR-PFS} entails the highest overheads
due to costly I/O flushes.

%\vspace{-2.0em}
\textbf{Performance of Demand Checkpointing }
%
%%%First, we analyze how $|CH|$ impacts the performance of recovering a process from its last demand
%%%checkpoint (see Figure~\ref{fig:nas-rec}).
%%%We run the NAS benchmark 10 times and after every such iteration we communicate the checksum
%%%necessary to recover the process. We use the \smalltt{nodes} t-awareness and compare 
%%%\smalltt{no-FT}, \smalltt{f-12.5-nodes} ($|CH|=12.5\% |CM|$), and \smalltt{f-6.25-nodes} ($|CH|=6.25\% |CM|$).
%%%%
%%%RMA's direct memory accesses ensure transparency and 
%%%relatively small overheads: when $|CH|=12.5\% |CM|$
%%%10 checksum transfers during 10 iterations make the run only 60\% slower than \smalltt{no-FT}. 
%%%%
%%%Second, 
We now analyze how the size of the log
impacts the number of demand checkpoints and the performance of fault-free runs (see Figure~\ref{fig:dchecks}).
Dedicating less than 44 MiB of memory for storing logs (per process) triggers demand checkpoint requests to
clear the log. This results in performance penalties but leaves more memory available to the user.

Second, we illustrate how $|CH|$ impacts the performance of recovering a process from its last demand
checkpoint (see Figure~\ref{fig:nas-rec}).
We run the NAS benchmark 10 times and after every such iteration we communicate the checksum
necessary to recover the process. We use the \smalltt{nodes} t-awareness and compare 
\smalltt{no-FT}, \smalltt{f-12.5-nodes} ($|CH|=12.5\% |CM|$), and \smalltt{f-6.25-nodes} ($|CH|=6.25\% |CM|$).
RMA's direct memory accesses ensure transparency and 
relatively small overheads: when $|CH|=12.5\% |CM|$
10 checksum transfers during 10 iterations make the run only 60\% slower than \smalltt{no-FT}. 
%

%\vspace{-0.5em}
\begin{figure}[h!] \centering
\includegraphics[width=0.3\textwidth]{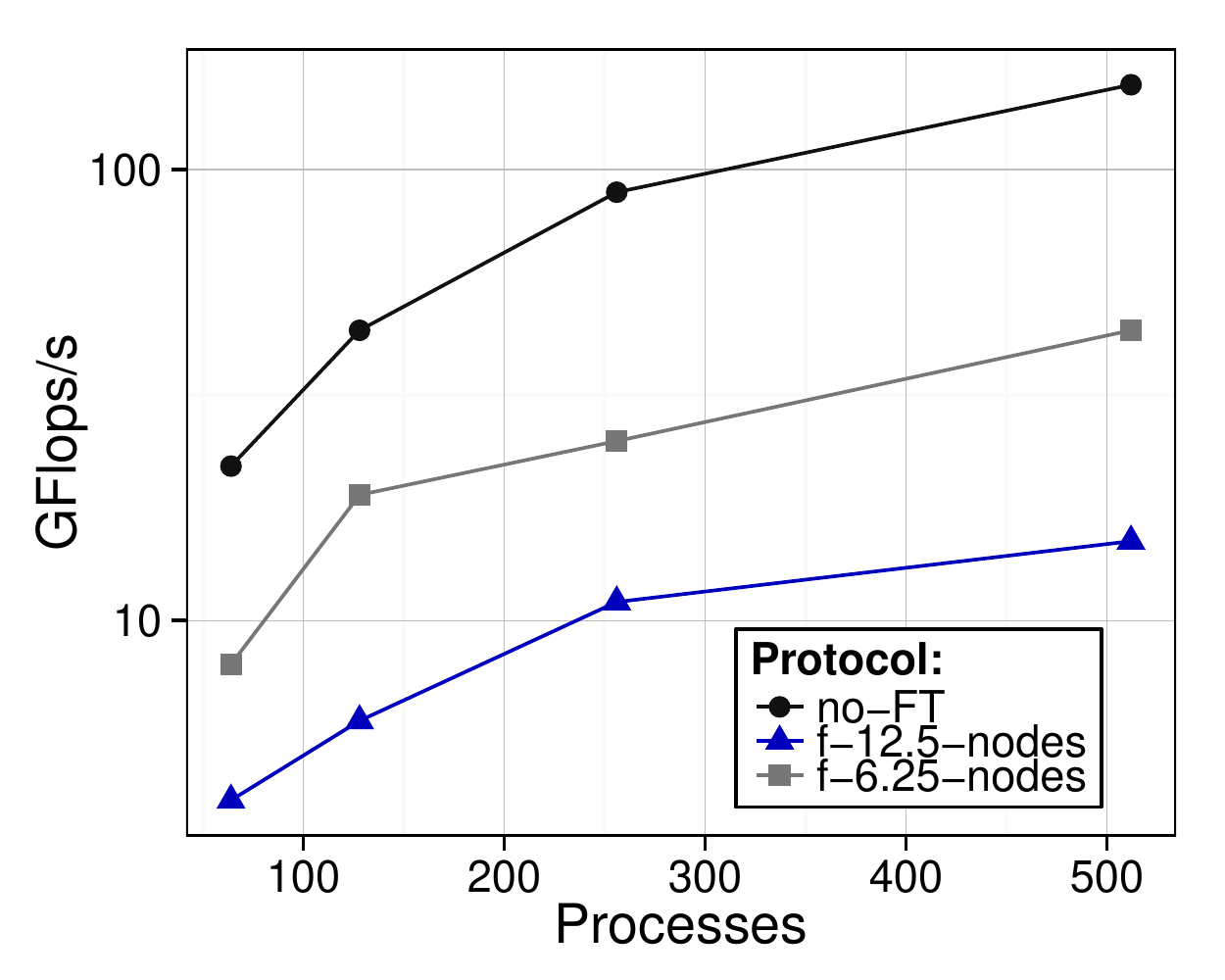}
\vspace{-0.5em}
\caption{NAS FFT (class C) recovery from a demand checkpoint. (\cref{sec:nas_eval}).} 
\label{fig:nas-rec}
%\vspace{-1.0em}
\end{figure}

%As FFT code uses very intensive communication patterns, our protocol
%degenerates automatically to the coordinated scheme after the first
%iteration.
%%%In this scenario, our ``protocol switch'' automatically turns off logging and
%%%degenerates the protocol to the coordinated scheme after the first
%%%iteration. Thus, we use this benchmark to analyze the overheads of our
%%%checkpointing protocol and compare it to SCR.
%
%
%As FFT code uses very intensive communication patterns, our protocol
%degenerates automatically to the coordinated scheme after the first
%iteration.

\goal{+ Describe the performance of fault-free runs when logging}

\textbf{Performance of Access Logging }
As the next step we evaluate our logging scheme.
Figure~\ref{fig:nas-logs} illustrates the performance of fault-free
runs. We compare
\smalltt{no-FT}, \flib{}, and our ML protocol (\smalltt{ML}).
\flib{} adds only $\approx$8-9\% of overhead to the baseline (\smalltt{no-FT})
and  consistently outperforms \smalltt{ML} by $\approx$9\%
due to the smaller amount of protocol-specific interaction between processes.

\goal{+ Describe how topo-awareness and |CH| influences results}

%\vspace{-2.0em}
\textbf{Varying |CH| and T-Awareness Policies }
Here, we analyze how $|CH|$ and t-awareness impact the performance of NAS FFT fault-free runs. We set $|CH|=12.5\% |CM|$ and $|CH|=6.25\% |CM|$, and we use the \smalltt{no-topo} and \smalltt{nodes} t-awareness policies. The results show that all these schemes differ negligibly from \smalltt{no-FT} by 1-5\%.% plots are included in\footnotemark[1]. 
%

%

%We also study the influence of $|CH|$ on the performance of process recovery.
%We run the NAS benchmark 10 times and after every such iteration we communicate the checksum
%necessary to recover one process. RMA's direct memory accesses ensure transparency and 
%relatively small overheads: when $|CH|=12.5\% |CM|$ (\smalltt{f-12.5-nodes})
%a process recovery is only 60\% slower than the fault-free \smalltt{no-FT}
%scenario.

%\goal{+ Describe the recovery benchmark}
%
%\textbf{Performance of Process Recovery }
%%
%Figure~\ref{fig:FFTRecoveryTime} shows how t-awa

%\vspace{-2.0em}
%%%\textbf{Performance of Process Recovery }
%%%%
%%%Figure~\ref{fig:FFTRecoveryTime} shows the performance of process recovery. 
%%%In this experiment we run the NAS benchmark 10 times and after every such iteration we perform the recovery. The
%%%protocol degenerates at the beginning, thus the recovery is global and
%%%spans over all the processes. We compare \smalltt{no-FT} to two scenarios:
%%%$|CH|=12.5\% |CM|$ (\smalltt{f-12.5-nodes}) and $|CH|=6.25\% |CM|$ (\smalltt{f-6.25-nodes}).
%%%For both we utilize the \smalltt{nodes} topology-awareness. Thanks to
%%%the utilization of RMA the recovery is both transparent and fast: for
%%%\smalltt{f-12.5-nodes} full recovery is only 60\% slower than the fault-free \smalltt{no-FT}
%%%scenario.

 \begin{figure*}[ht]
\centering
\scalebox{0.55}{\input{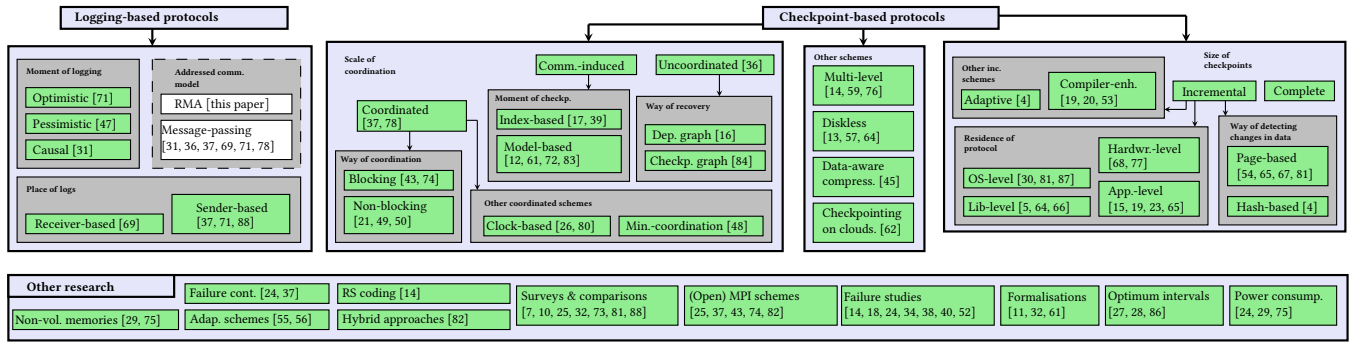}}
%\vspace{-1.5em}
\caption{An overview of existing checkpointing and logging schemes (\cref{sec:relatedWork}). A dashed rectangle illustrates a new sub-hierarchy introduced in the paper: dividing the logging protocols with respect to the \emph{communication model} that they address.}
\label{fig:overview}
%\vspace{-1.0em}
\end{figure*}

%\vspace{-1.0em}
\subsubsection{Key-Value Store}
\label{sec:kv_eval}

\goal{+ Describe our DHT}

Our key-value store is based on a simple distributed hashtable (DHT)
that stores 8--Byte integers.
The DHT consists of parts called \emph{local volumes} constructed with
fixed-sized arrays. Every local volume is managed
by a different process. Inserts are
based on MPI-3 atomic Compare-And-Swap and Fetch-And-Op functions.
Elements after hash collisions are inserted in the overflow heap that is the part of each local volume.
To insert an element, a thread atomically updates the pointers to the next free cell and the last element
in the local volume.
Memory consistency is ensured with flushes. One \textsc{get} and one \textsc{put} are logged if there is no hash collision, otherwise 6 \textsc{put}s and 4 \textsc{get}s are recorded.

%%%Our key-value store implementation is based on a simple distributed hashtable (DHT).
%%%Each process manages a part of the DHT
%%%called the \emph{local volume} consisting of a table of elements and an
%%%overflow heap to store elements after collisions. The table
%%%and the heap are constructed using fixed-size arrays. In order to avoid
%%%traversing the arrays, pointers to most recently inserted items as
%%%well as to the next free cells are stored along with the remaining data
%%%in each local volume. The elements of the hashtable are 8--Byte
%%%integers. Table and overflow list constitute an MPI window. Inserts are
%%%based on MPI-3 atomic Compare-And-Swap (CAS) and Fetch-And-Op operations. If a collision happens, the loosing thread
%%%acquires a new element in the overflow list by atomically incrementing
%%%the next free pointer. It also updates the last pointer using a second
%%%CAS. Flushes are used to ensure memory consistency. If there is no hash collision, one \textsc{get} and one \textsc{put} will be logged. If a collision happens, then there are in sum 6 \textsc{put}s and 4 \textsc{get}s to be recorded. 

\goal{+ Describe the DHT fault-free-runs benchmark}

%\vspace{-1.8em}
\textbf{Performance of Access Logging }
We now measure the relative performance
penalty of logging \textsc{put}s and \textsc{get}s. During the
benchmark, processes insert random elements with random keys. We focus
on inserts only as they are perfectly representative for the logging
evaluation. To simulate realistic requests, every process waits
for a random time after every insert. The function that we use to
calculate this interval is based on the exponential probability
distribution: $f \delta e^{-\delta x}$, where $f$ is a scaling factor,
$\delta$ is a rate parameter and $x \in [0;b)$ is a random number.
The selected parameter values ensure that
  every process spends
  $\approx$5-10\% of the total runtime on inserting elements. For many
  computation-intense applications this is already a high amount of
  communication. We again compare
  \smalltt{no-FT}, \smalltt{ML}, and two \flib{} scenarios: \smalltt{f-puts} (logging
  only \textsc{put}s) and \smalltt{f-puts-gets} (logging \textsc{put}s
  and \textsc{get}s). We fix $|CH| = 12.5\% |CM|$ and use the
  \smalltt{nodes} t-awareness. We skip SCR as it
  does not enable logging.
%%The goal of this experiment is to measure the relative performance
%%penalty of logging \textsc{put}s and \textsc{get}s. During the
%%benchmark, processes insert random elements with random keys. We focus
%%on inserts only as they are perfectly representative for the logging
%%evaluation. In order to simulate realistic requests, every process waits
%%for a random time after every insert. The function that we use to
%%calculate this interval is based on the exponential probability
%%distribution: $f \delta e^{-\delta x}$, where $f$ is a scaling factor,
%%$\delta$ is a rate parameter and $x$ is a random number from the
%%interval $[0;b)$. We fix the parameters in such a way that the time that
%%  every process spends only performing inserts amounts to
%%  $\mathtt{\sim}$ 10\% of the total running time. For many
%%  computation-intense applications this is already a high amount of
%%  communication. We compare the variant with no fault-tolerance
%%  (\smalltt{no-FT}) and two \flib{} scenarios: \smalltt{f-puts} (logging
%%  only \textsc{put}s) and \smalltt{f-puts-gets} (logging \textsc{put}s
%%  and \textsf{get}s). We again fix $|CH| = 12.5\% |CM|$ and we use the
%%  \smalltt{nodes} topology-awareness. We do not compare to SCR as it
%%  does not support logging.

%We devote 4 MB of memory space for logs for every process. 

\goal{+ Describe the benchmark results}

We present the results in Figure~\ref{fig:hashFFrun}. For $N=256$, the logging overhead over the baseline (\smalltt{no-FT}) is: $\approx$12\% (\smalltt{f-puts}), 33\% (\smalltt{f-gets}), and 40\% (\smalltt{ML}). The overhead of logging \textsc{put}s in is due to the fact that every operation is recorded directly after issuing. Traditional message passing protocols suffer from a similar effect~\cite{Elnozahy:2002:SRP:568522.568525}. The overhead generated by logging \textsc{get}s in \smalltt{f-puts-gets} and \smalltt{ML} is more significant because, due to RMA's one-sided semantics, every \textsc{get} has to be recorded {remotely}. 
In addition, \smalltt{f-puts-gets} suffers from synchronization overheads (caused by concurrent accesses to $LG$), while \smalltt{ML} from inter-process protocol-specific communication.
Discussed overheads heavily depend on the application type. Our key-value store constitutes a worst-case scenario
because it does not allow for long epochs that could enable, e.g., sending the logs of multiple \textsc{get}s in a bulk. 
The performance penalties would be smaller in applications that
overlap computation with communication and use non blocking \textsc{get}s.

%Finally, maintaining all required counters also generates additional traffic in the network. 

%\vspace{-0.5em}
\section{Related Work}
\label{sec:relatedWork}

\goal{Introduce the section and imply there's no research on resilience in RMA}

In this section we discuss existing checkpointing and
logging schemes (see Figure~\ref{fig:overview}). For excellent surveys, see
\cite{Elnozahy:2002:SRP:568522.568525,
Alvisi:1998:MLP:630821.631222,
vasavada2011comparing}. Existing work on fault tolerance in RMA/PGAS is scarce,
an example scheme that uses PGAS for data replication can be found in~\cite{5738978}.

%Figure~\ref{fig:overview} illustrates all the schemes and how they are
%categorized.

% We discuss how we combine and go beyond
%their chosen features to construct our protocols at the end of this section.

%\vspace{-0.8em}
\subsection{Checkpointing Protocols}

%%When taking into account the scale of coordination, checkpoint based methods are traditionally divided into three main categories: \emph{uncoordinated}, \emph{coordinated}, and \emph{communication induced}. When considering the size of checkpoints, there are \emph{complete} and \emph{incremental} protocols.

\goal{+ Say how we divide checkpointing schemes}

These schemes are traditionally divided into \emph{uncoordinated}, \emph{coordinated}, and \emph{communication induced}, depending on process coordination scale~\cite{Elnozahy:2002:SRP:568522.568525}. There are also \emph{complete} and \emph{incremental} protocols that differ in checkpoint sizes~\cite{vasavada2011comparing}.

%\vspace{-2.2em}
%\paragraph{Uncoordinated Schemes}

\textbf{Uncoordinated Schemes }
Uncoordinated schemes do not synchronize while checkpointing, but may suffer from \emph{domino effect} or complex recoveries~\cite{Elnozahy:2002:SRP:568522.568525}. Example protocols are based on \emph{dependency}~\cite{Bhargava25775} or \emph{checkpoint graphs}~\cite{Elnozahy:2002:SRP:568522.568525}. A recent scheme targeting large-scale systems is Ken~\cite{Yoo:2012:CRA:2342821.2342824}.
%%Uncoordinated schemes do not synchronize while checkpointing, but they need to record dependencies between consecutive checkpoints and may suffer from \emph{domino effect}, \emph{orphan processes} or storing many checkpoints~\cite{Elnozahy:2002:SRP:568522.568525}. Several methods for rollback recovery were proposed, e.g., an approach based on \emph{dependency graph}~\cite{Bhargava25775} or \emph{checkpoint graph}~\cite{Elnozahy:2002:SRP:568522.568525}. Example recent protocol targeting large-scale systems is Ken~\cite{Yoo:2012:CRA:2342821.2342824}.

%%\vspace{-2.2em}
%\paragraph{Coordinated Schemes}

\textbf{Coordinated Schemes }
Here, processes synchronize to produce consistent global checkpoints. There is no domino effect and recovery is simple but synchronization may incur severe overheads. Coordinated schemes can be \emph{blocking}~\cite{Elnozahy:2002:SRP:568522.568525} or \emph{non-blocking}~\cite{Chandy:1985:DSD:214451.214456}. There are also schemes based on \emph{loosely synchronized clocks}~\cite{Tong:1992:RRD:628900.629082} and \emph{minimal coordination}~\cite{1702129}.
%%In coordinated schemes processes synchronize in order to produce consistent global checkpoints. There is no domino effect and only one checkpoint has to be stored. Still, system execution may suffer from the additional synchronization. Traditionally, coordinated schemes are divided into \emph{blocking}~\cite{Elnozahy:2002:SRP:568522.568525} and \emph{non-blocking}~\cite{Chandy:1985:DSD:214451.214456}. There are also schemes based on \emph{loosely synchronized clocks}~\cite{Tong:1992:RRD:628900.629082} and \emph{minimal coordination}~\cite{1702129}.

%%\vspace{-2.0em}
%\paragraph{Communication Induced Schemes}

\textbf{Communication Induced Schemes }
Here, senders add scheme-specific data to application messages that receivers use to, e.g., avoid taking useless checkpoints. These schemes can be \emph{index-based}~\cite{632814} or \emph{model-based}~\cite{Elnozahy:2002:SRP:568522.568525,342127}.
%%In such protocols processes piggyback scheme-specific data on each application message; the receiver uses this information to, e.g., avoid creating useless checkpoints.~\cite{Elnozahy:2002:SRP:568522.568525} divides these protocols into \emph{index-based}~\cite{632814} and \emph{model-based}~\cite{342127}.

%%\vspace{-2.0em}
%\paragraph{Incremental Checkpointing}

\textbf{Incremental Checkpointing }
An incremental checkpoint updates only the data that changed since the previous checkpoint. These protocols are divided into page-based~\cite{vasavada2011comparing} and hash-based~\cite{Agarwal:2004:AIC:1006209.1006248}. They can reside at the level of an \emph{application}, a \emph{library}, an \emph{OS}, or \emph{hardware}~\cite{vasavada2011comparing}. Other schemes can be \emph{compiler-enhanced}~\cite{Bronevetsky:2008:CIC:1345206.1345253} or \emph{adaptive}~\cite{Agarwal:2004:AIC:1006209.1006248}.
%%An incremental checkpoint updates only the data which changed since the previous checkpoint. Incremental protocols are categorized into page-based~\cite{vasavada2011comparing} and hash-based~\cite{Agarwal:2004:AIC:1006209.1006248}. Finally, incremental protocols can also be: \emph{application-level}, \emph{run-time library level}, \emph{operating system level}, and \emph{hardware level}~\cite{vasavada2011comparing}. There are also \emph{compiler-enhanced}~\cite{Bronevetsky:2008:CIC:1345206.1345253} and hybrid (both complete and incremental) systems~\cite{5695644}. \emph{Adaptive} variant is presented in~\cite{Agarwal:2004:AIC:1006209.1006248}.

%%\vspace{-2.0em}
%\paragraph{Others}

\textbf{Others }
Recently, \emph{multi-level} checkpointing was introduced~\cite{Moody:2010:DME:1884643.1884666,Bautista-Gomez:2011:FHP:2063384.2063427,Sato:2012:DMN:2388996.2389022}. \emph{Adaptive} checkpointing based on failure prediction is discussed in~\cite{li2007using}. A study on checkpointing targeted specifically at GPU-based computations can be found in~\cite{6012895}. \cite{730527} presents diskless checkpointing. Other interesting schemes are based on: Reed-Solomon coding~\cite{Bautista-Gomez:2011:FHP:2063384.2063427}, cutoff and compression to reduce checkpoint sizes~\cite{6264674}, checkpointing on clouds~\cite{Nicolae:2011:BEC:2063384.2063429}, reducing I/O bottlenecks~\cite{scc}, and performant checkpoints to PFS~\cite{Arteaga:2011:TSA:2060102.2060540}.
%
%in distributed RAID systems

%%More recently, \emph{multi-level} checkpoint schemes were introduced~\cite{Moody:2010:DME:1884643.1884666,Bautista-Gomez:2011:FHP:2063384.2063427,Sato:2012:DMN:2388996.2389022}. \emph{Adaptive} checkpointing based on failure prediction is discussed in~\cite{li2007using}. A study on checkpointing targeted specifically at GPU-based computations can be found in~\cite{6012895}.~\cite{730527} present diskless checkpointing. Other interesting schemes are based on: Reed-Solomon coding~\cite{Bautista-Gomez:2011:FHP:2063384.2063427}, cutoff and compression schemes that reduce checkpoint sizes~\cite{6264674}, checkpointing on clouds~\cite{Nicolae:2011:BEC:2063384.2063429}, reducing I/O bottlenecks in distributed RAID systems~\cite{scc}, and improving the performance of checkpoints to PFS~\cite{Arteaga:2011:TSA:2060102.2060540}.

%A study on checkpointing targeted specifically at GPU-based computations can be found in~\cite{6012895}.~

%\vspace{-1.0em}
\subsection{Logging Protocols}
\label{sec:loggingProtocols}

\goal{+ Say how we divide logging schemes}

%%These schemes usually combine checkpointing and logging to enable processes to replay their execution after a failure beyond the most recent checkpoint. Log-based protocols are traditionally categorized into: \emph{pessimistic}, \emph{optimistic}, \emph{causal}~\cite{Elnozahy:2002:SRP:568522.568525}. Moreover, there are also \emph{sender-based}~\cite{Rao98hybridmessage,Riesen:2012:ASI:2388996.2389021,6012907} and \emph{receiver-based}~\cite{Elnozahy:2002:SRP:568522.568525} protocols.
Logging enables restored processes to replay their execution beyond the most recent checkpoint. Log-based protocols are traditionally categorized into: \emph{pessimistic}, \emph{optimistic}, \emph{causal}~\cite{Elnozahy:2002:SRP:568522.568525}; they can also be \emph{sender-based}~\cite{Riesen:2012:ASI:2388996.2389021,6012907} and \emph{receiver-based}~\cite{Elnozahy:2002:SRP:568522.568525} depending on which side logs messages.

%%\vspace{-2.2em}
%\paragraph{Pessimistic Schemes}

\textbf{Pessimistic Schemes }
Such protocols log events before they influence the system. This ensures no orphan processes and simpler recovery, but may incur severe overheads during fault-free runs. An example protocol is V-MPICH~\cite{1592865}.
%%Pessimistic protocols record each event before it actually influences the system execution. Such an approach guarantees, among others: no orphan processes and simpler recovery. However, significant overhead during fault-free runs is the major drawback of these protocols. The overhead varies depending on the type of storage utilized. Example systemis V-MPICH~\cite{1592865}.
%~\cite{Rao98hybridmessage,1592865}.

%%\vspace{-2.0em}
%\paragraph{Optimistic Schemes}

\textbf{Optimistic Schemes }
Here, processes postpone logging messages to achieve, e.g., better computation-communication overlap. However, the algorithms for recovery are usually more complicated and crashed processes may become orphans~\cite{Elnozahy:2002:SRP:568522.568525}. A recent scheme can be found in~\cite{Riesen:2012:ASI:2388996.2389021}.
%%In this family of protocols, processes keep the logs in their volatile memory and flush the data at some intervals. Better fault-free performance is an obvious advantage of these protocols. On the other hand, the algorithms for recovery and garbage collection are usually more complicated and (after a failure) some processes may become orphans~\cite{Elnozahy:2002:SRP:568522.568525}. A more recent implementation of an optimistic logging protocol can be found in~\cite{Riesen:2012:ASI:2388996.2389021}.

%%\vspace{-2.0em}
%\paragraph{Causal Schemes}

\textbf{Causal Schemes }
In such schemes processes log and exchange (by piggybacking to messages) dependencies needed for recovery. This ensures no orphans but may reduce bandwidth~\cite{Elnozahy:2002:SRP:568522.568525}. An example protocol is discussed in~\cite{142678}.
%%Causal protocols record the dependency information needed for restoring system to a consistent state after a failure. Representations of dependency graphs are piggybacked on the messages. This guarantees no orphans in the system. Still, the system execution may suffer from the often large amount of additional information in each message.\cite{Elnozahy:2002:SRP:568522.568525}. Example implementations of a causal protocols are presented in~\cite{142678,Rao98hybridmessage}.

%\vspace{-1.0em}
\subsection{Other Important Studies \& Discussion}

\goal{+ Describe other related research}

%%Determining the optimum interval in coordinated protocols was addressed in, e.g.,~\cite{Daly:2006:HOE:1134241.1134248}. Formalizations approaching the issue of consistent global snapshots can be found in, e.g.,~\cite{342127,Elnozahy:2002:SRP:568522.568525}. Power consumption was addressed in~\cite{Sardashti:2012:ULN:2304576.2304587,  Chung:2012:CDS:2388996.2389075}. \emph{Containment domains} for encapsulating failures within a pre-determined hierarchical scope are discussed in~\cite{Chung:2012:CDS:2388996.2389075}. Modeling, classification and prediction of failures was addressed in~\cite{Bautista-Gomez:2011:FHP:2063384.2063427, Chung:2012:CDS:2388996.2389075}. Work on send determinism in MP can be found in~\cite{6012907}.
Deriving an optimum checkpointing interval is presented in~\cite{Daly:2006:HOE:1134241.1134248}. Formalizations targeting resilience can be found in~\cite{342127,Elnozahy:2002:SRP:568522.568525}. Power consumption was addressed in~\cite{Sardashti:2012:ULN:2304576.2304587,Chung:2012:CDS:2388996.2389075}. \emph{Containment domains} for encapsulating failures within a hierarchical scope are discussed in~\cite{Chung:2012:CDS:2388996.2389075}. Modeling and prediction of failures is addressed in~\cite{Bautista-Gomez:2011:FHP:2063384.2063427, Chung:2012:CDS:2388996.2389075}. Work on send determinism in MP can be found in~\cite{6012907}.

\goal{+ Say why we differ \& are better than the above}

Our study goes beyond the existing research scope presented in this
section. First, we develop a fault tolerance model that covers virtually whole rich RMA
semantics. Other existing formalizations (e.g.,~\cite{342127,Elnozahy:2002:SRP:568522.568525,
Alvisi:1998:MLP:630821.631222}) target MP only.
We then use the model to formally analyze why resilience for RMA
differs from MP and to design checkpointing, logging, and recovery
protocols for RMA. We identify and propose solutions to several
challenges in resilience for RMA that \emph{do not} exist in MP, e.g.:
consistency problems
caused by the relaxed RMA memory model (\cref{sec:taking_coordinated_ckp}, \cref{sec:taking_uncoordinated_ckp}, \cref{sec:transparentLogging}), access non-determinism (\cref{sec:managing_unc}), issues due to
one-sided RMA communication (\cref{sec:rma_vs_mp_ucc}), logging multiple RMA-specific orders (\cref{sec:logging_order_info}), etc.
Our model enables proving correctness of proposed schemes.
Extending our model for arbitrary hardware hierarchies generalizes the approach from~\cite{Bautista-Gomez:2011:FHP:2063384.2063427}
and enables formal reasoning about crashes of hardware elements and process distribution.
Finally, our protocol leverages and combines several important concepts and mechanisms (Daly's interval~\cite{Daly:2006:HOE:1134241.1134248}, multi-level design~\cite{Moody:2010:DME:1884643.1884666}, etc.) to improve the resilience of RMA systems even further and is the first implementation of holistic fault tolerance for RMA.

\section{Conclusion}

\goal{State RMA is becoming popular but there's no fault-tolerance for it}

RMA programming models are growing in popularity and importance as they allow for the best utilization of hardware features such as OS-bypass or zero-copy data transfer. Still, little work addresses fault tolerance for RMA.

\goal{Advertise our formal model and describe its broader applications}

We established, described, and explored a
complete formal model of fault tolerance for RMA and illustrated how to use it to design and reason about resilience protocols running on flat and hierarchical machines. It will play an important role in making emerging RMA programming fault tolerant and can be easily extended to cover, e.g., stable storage.

\goal{Describe our protocol/implementation and suggest its broader applications}

Our study does not resort to traditional less scalable mechanisms that often rely on costly I/O flushes.
The implementation of our holistic protocol adds negligible overheads to the applications runtime, for example 1-5\% for in-memory checkpointing and 8\% for fully transparent logging of remote memory accesses in the NAS 3D FFT code. Our probability study shows that the protocol offers high resilience. The idea of demand checkpoints will help alleviate the problem of limited memory amounts in today's petascale and future exascale computing centers. 

\goal{Describe broader potential behind the whole paper}

Finally, our work provides the basis for
further reasoning about fault-tolerance not only for RMA,
but also for all the other models that can be constructed upon it, such
as task-based programming models. This will play an important role in complex heterogeneous large-scale systems.

{
\vspace{0em}\section*{Acknowledgements}
%\vspace{-0.5em}
We thank the CSCS team granting access to the Monte Rosa machine, and for their
excellent technical support. We thank Franck Cappello for inspiring remarks.
We thank Timo Schneider
for his immense help with computing infrastructure at SPCL.}

\bibliographystyle{abbrv}
\bibliography{references}

\end{document}